%% file: draft1.tex
\documentclass[english,10pt]{article}

\usepackage[english]{babel}
\usepackage[utf8x]{inputenc}
\usepackage[T1]{fontenc}
\usepackage[formats]{listings}
\usepackage{geometry}
\usepackage{enumitem}

\usepackage{graphicx}
\usepackage[colorinlistoftodos]{todonotes}
\usepackage{setspace}
\usepackage{placeins}
\usepackage{enumitem}
\usepackage{titlesec}
\usepackage{cite}
\usepackage{comment}
\usepackage{caption}
\usepackage{subcaption}
\usepackage{enumitem}

\usepackage[colorlinks=true, allcolors=blue]{hyperref}

\usepackage{amsmath}
\usepackage{amsthm}
\usepackage{amssymb}
\usepackage{amsfonts}
\usepackage{bbm}
\usepackage{bm}
\usepackage{nicefrac}
\usepackage{mathtools}

\usepackage{tikz}

\usepackage{graphicx}
\usepackage{fullpage}
\usepackage{float}
\usepackage{wrapfig}
\usepackage{caption}

\usepackage{algpseudocode}
\usepackage{algorithm}
\usepackage{listings}

\usepackage{amsthm}
\usepackage{dsfont}
\usepackage{array}
\usepackage{mathrsfs}
\usepackage{cite}
\usepackage{comment}
\usepackage{mathrsfs}

\makeatother

\usepackage{babel}
\usepackage{color}
\usepackage{centernot}

\usepackage{babel}
\usepackage{sansmath}

\usepackage{custom}

\usepackage[colorinlistoftodos]{todonotes}

\title{Contextual Stochastic Block Model: Sharp Thresholds and Contiguity}

\author{Chen Lu\thanks{Department of Mathematics, Massachusetts Institute of Technology, Cambridge MA 02139, U.S.A.} \and 
  Subhabrata Sen\thanks{Department of Statistics, Harvard
    University, Cambridge, MA 02138, U.S.A.} }

\theoremstyle{plain}\newtheorem{lemma}{\textbf{Lemma}}\newtheorem{theorem}{\textbf{Theorem}}\newtheorem{definition}{\textbf{Definition}}\newtheorem{proposition}{\textbf{Proposition}}

\theoremstyle{definition}

\theoremstyle{definition}\newtheorem{remark}{\textbf{Remark}}

\makeatletter
\renewenvironment{proof}[1][\proofname] {
	\par\pushQED{\qed}\normalfont
	\topsep6\p@\@plus6\p@\relax
	\trivlist\item[\hskip\labelsep\bfseries#1\@addpunct{:}]
 	\ignorespaces
} {
	\popQED\endtrivlist\@endpefalse
}
\makeatother

\begin{document}

\maketitle



\begin{abstract}
We study community detection in the \emph{contextual stochastic block model} \cite{yan2020covariate,deshpande2018contextual}. In \cite{deshpande2018contextual}, the second author studied this problem in the setting of sparse graphs with high-dimensional node-covariates. Using the non-rigorous \emph{cavity method} from statistical physics \cite{mezard2009information}, they conjectured the sharp limits for community detection in this setting. Further, the information theoretic threshold was verified, assuming that the average degree of the observed graph is large. It is expected that the conjecture holds as soon as the average degree exceeds one, so that the graph has a giant component. We establish this conjecture, and characterize the sharp threshold for detection and weak recovery.
\end{abstract} 

\section{Introduction}
The community detection problem arises routinely in diverse applications, and has received significant attention recently in statistics and machine learning. In the simplest version of this problem, given access to a graph, one seeks to cluster the vertices into interpretable \emph{communities} or groups of vertices, which are believed to reflect latent similarities among the nodes. From a theoretical standpoint, this problem has been extensively analyzed under specific generative assumptions on the observed graph; the most popular generative model in this context is the \emph{stochastic block model} (SBM) \cite{holland1983stochastic}. Inspired by intriguing conjectures arising from the statistical physics community \cite{krzakala2013spectral}, community detection under the stochastic block model has been studied extensively. As a consequence, the precise information theoretic limits for recovering the underlying communities have been derived, and optimal algorithms have been identified in this setting (for a survey of these recent breakthroughs, see \cite{abbe2017community}). 

In reality, the practitioner often has access to additional information in the form of node covariates, which complements the graph information. Statistically, it is natural to believe that clustering performance can be significantly improved by combining this covariate information with the graph structure. However, establishing this improvement in a formal context, and deriving procedures which combine the two are not straightforward. In \cite{yan2020covariate}, the authors formalize this question, and introduce a simple model for community detection with node covariates. We use the same framework in this paper. 

We observe a graph $G = (V, E)$ on $n$ vertices drawn from the sparse Stochastic Block Model $G(n, \frac{a}{n}, \frac{b}{n})$. Formally, we sample a community assignment vector $\sigma \in \{ \pm 1\}^n$ uniformly; given $\sigma$, we draw edges with probability
\begin{align}
\mathbb{P}[\{i, j\} \in E] = \begin{cases}
\frac{a}{n} \quad \textrm{if}\,\,\,\, \sigma_i = \sigma_j, \\
\frac{b}{n} \quad \textrm{o.w.} 
\end{cases} \nonumber 
\end{align}
Let $\mathbf{A} = (A_{ij}) \in \mathbb{R}^{n \times n}$ denote the adjacency matrix of the graph $G$. We define the average degree parameter $d = \frac{a+b}{2}$ and parametrize $a = d + \lambda \sqrt{d}$ and $b=d - \lambda \sqrt{d}$.

Further, at each node of the graph $G$, we observe a $p$-dimensional vector of covariates $\{B_i : 1\leq i \leq n\}$. The covariates are also correlated with the underlying community assignment. Specifically, we observe 
\begin{align}
B_i = \sqrt{\frac{\mu}{n}} \sigma_i u + Z_i, \nonumber  
\end{align}
where $u \sim \mathcal{N}(0, I_p)$ is a latent gaussian vector, and $Z_i \sim \mathcal{N}(0,I_p)$. We construct the matrix $\mathbf{B} = [B_1, \cdots, B_n ]^\top \in \mathbb{R}^{n \times p}$. In \cite{yan2020covariate}, the authors introduce a semidefinite programming (SDP) based algorithm for community detection in the above setting, which combines the graph with the node covariates. However, their results do not identify the information theoretic limits of community detection in this context, and do not identify the optimal community detection algorithm in this setting.

Note that when one has access to either the graph information or the covariate information, the information theoretic threshold is well known. Under the parametrization of the SBM, where only $\mb A$ is given, detection of the underlying community structure, as well as non-trivial recovery, are possible if and only if $\lambda >1$. On the other hand, the case when only the covariate information $\mb B$ corresponds to a Gaussian mixture clustering problem. Under a high-dimensional asymptotic regime $\frac{n}{p} \to \gamma \in (0,\infty)$, random matrix considerations based on the BBP phase transition \cite{baik2005phase}, and contiguity arguments based on the second moment method \cite{perry2018optimality} imply that non-trivial detection and recovery are possible if and only if $\mu^2 > \gamma$.

In \cite{deshpande2018contextual}, the second author and coauthors studied detection and recovery under this model in a high-dimensional asymptotic regime $\frac{n}{p} \to \gamma \in (0, \infty)$ \cite{deshpande2018contextual}, and 
conjectured the sharp information theoretic limits in this problem: the underlying community structure can be detected if and only if $\lambda^2 + \frac{\mu^2}{\gamma} >1$. In particular, this suggests that upon combining the graph information appropriately with the covariates, it is statistically possible to improve upon the optimal performance based on any single information source. The conjecture was derived using the  non-rigorous \emph{cavity method} from statistical physics \cite{mezard2009information}, and rigorously established under an additional high-degree asymptotic $d \to \infty$ (after $n \to \infty$). However, numerical experiments in \cite{deshpande2018contextual} suggest that this high-degree asymptotic is unnecessary, and that the results are true as soon as $d >1$. In this paper, we formally establish this conjecture. 
Throughout this article, we will work under the same high-dimensional asymptotic $\frac n p \to \gamma \in (0,\infty)$.

Our main contributions in this article are as follows. 
\begin{itemize}
\item[(i)] We first examine the detection problem \eqref{eq:detection}, and establish the sharp threshold for mutual contiguity (see Theorem \ref{thm:it_threshold}). The testing lower bound is derived using a traditional second moment argument. The upper bound is significantly more challenging---we devise a test statistic by counting certain appropriate self-avoiding walks in the graph. 

\item[(ii)] Next, we turn to weak recovery, and establish that the threshold for weak recovery coincides with that for detection in this context. To establish the positive half of this result, we crucially utilize the self-avoiding walk based estimation idea introduced in \cite{massoulie2014community,hopkins2017bayesian}; however, the presence of the graph data with the covariates makes this application significantly more challenging. We devise estimates for the pairwise correlations of the memberships by counting appropriate walk based statistics, and then perform weak recovery by a subsequent projection and rounding step (we refer to Section \ref{sec:ubd} for further details). This is one of the main technical contributions of this article. 

\item[(iii)] We then turn to the contiguity regime, and derive a precise expansion of the likelihood ratio in the contiguity regime in terms of appropriate cycle statistics. In turn, this identifies the precise statistics which distinguish the null from the alternative. 
\end{itemize}

\subsection{Main Results} 
Consider the hypothesis testing problem 
\begin{align}
\mathrm{H}_0: (\lambda, \mu) = (0,0) \quad \textrm{vs.} \quad \mathrm{H}_1: (\lambda, \mu) \neq (0,0). \label{eq:detection}
\end{align}
We will denote the joint distributions of the data $(\mb A, \mb B)$ by $\mathbb{P}_{\lambda, \mu}$, and keep the dependence on $n,p$ implicit throughout. Further, we will assume that we are above the threshold for emergence of the giant component, i.e. $d>1$ \cite{bollobas2007phase,janson2011random}.

\begin{theorem}[Detection]
\label{thm:it_threshold} 
If $\lambda^2 + \frac{\mu^2}{\gamma} < 1$, $\mathbb{P}_{\lambda,\mu}$ is contiguous to $\mathbb{P}_{0,0}$. On the other hand, if $\lambda^2 + \frac{\mu^2}{\gamma}>1$, the sequences $\mathbb{P}_{0,0}$ and $\mathbb{P}_{\lambda, \mu}$ are mutually asymptotically singular. 
\end{theorem}

\noindent
Our next result addresses the threshold for weak recovery. We recall the relevant notion of weak recovery in this context. 
\begin{definition}[Weak Recovery] An estimator $\hat{\sigma}:= \hat{\sigma}(\mb A, \mb B) \in [-1,1]^n$ achieves weak recovery if there exists $\varepsilon>0$, independent of $n$, such that 
\begin{align}
\frac{1}{n}\mathbb{E}_{\lambda, \mu}[|\langle \sigma, \hat{\sigma} \rangle|] \geq \varepsilon \nonumber 
\end{align}
as $n \to \infty$. We say that weak recovery is possible if there exists an estimator $\hat{\sigma}$ which achieves weak recovery. 

\end{definition}
\begin{theorem}[Weak Recovery]
\label{thm:recovery} 
If $\lambda^2 + \frac{\mu^2}{\gamma} < 1$, then weak recovery is impossible. On the other hand, weak recovery is possible when $\lambda^2 + \frac{\mu^2}{\gamma} > 1$. 
\end{theorem}

Finally, we turn to the contiguity phase $\lambda^2 + \frac{\mu^2}{\gamma} < 1$, and derive an expansion for the likelihood ratio. We denote the likelihood ratio by 
\begin{align}
L_n = \frac{\mathrm{d} \mathbb{P}_{\lambda, \mu} }{\mathrm{d} \mathbb{P}_{0,0}}. \nonumber 
\end{align}
 In this regime, $\mathbb{P}_{\lambda, \mu}$ and $\mathbb{P}_{0,0}$ cannot be distinguished with asymptotically negligible Type I and Type II errors. Statistically, the natural problem of interest concerns optimal detection, which is achieved by the likelihood ratio test (LRT). We will derive asymptotic expansions of the likelihood ratio under the null and the alternative. As a consequence, we will obtain the optimal power of the LRT. Along the way, we will obtain a family of statistics which ``determine" the likelihood ratio. This will suggest computationally feasible statistics which attain optimal detection performance in this contiguous regime. Similar expansions for the likelihood ratio have been derived for \emph{pure} spiked gaussian problems (the model $B$ and its symmetric analogue) in the recent literature \cite{banerjee2018lr,alaoui2018detection,el2020fundamental,johnstone2020testing,onatski2013asymptotic,onatski2014signal}. Our approach in this regard will be motivated by the techniques introduced in \cite{banerjee2018lr}. However, we note that in contrast to this literature, we have \emph{both} a sparse random graph component, and a gaussian model. This necessitates crucial technical modifications---we emphasize the main differences in Section \ref{sec:lrtexpansion}.

To this end, let us first introduce a class of cycle statistics. We will denote $\omega$ as a cycle on the factor graph corresponding to the posterior distribution \cite{mezard2009information}, shown in Figure~\ref{fig:factor}. Specifically, the factor graph is denoted as $G_{\msf F} = (V_{\msf F},E_{\msf F})$. The vertices are split into two groups $V_{\msf F}= V_1 \cup V_2$, where $V_1$ denotes vertices from the adjacency matrix $\mb A$, with $|V_1|=n$, and $V_2$ denotes vertices from the covariate matrix $\mb B$, so $|V_2|=p$. In Figure~\ref{fig:factor}, vertices in $V_1$ are shown by dots, and those in $V_2$ are shown by squares. The edges also split into two groups $E_{\msf F}= E_1 \cup E_2$, where $E_1= \{ \{i_1, i_2\}: i_1, i_2 \in V_1\}$, and $E_2 = \{ \{ i,j\}: i \in V_1, j \in V_2\}$. We will refer to edges in $E_1$ as $A$ edges, and edges in $E_2$ as $B$ edges. Because $B$ edges must appear in consecutive pairs in a cycle, we refer to such pairs of $B$ edges as $B$ wedges. The graph of a cycle $\omega$ is denoted as $G_\omega = (V_\omega, E_\omega)$. We use $k$ to denote the length of the cycle, and $l$ to denote the number of $B$ wedges in the cycle. 
For a cycle $\omega$, we denote $G_{\omega,A} = (V_{\omega,A}, E_{\omega, A})$ the subgraph of the $A$ edges, and $G_{\omega,B} = (V_{\omega,B}, E_{\omega, B})$ is the subgraph of $B$ edges.

\begin{figure}[ht]
\begin{center}
\includegraphics[scale=0.8, angle=-90]{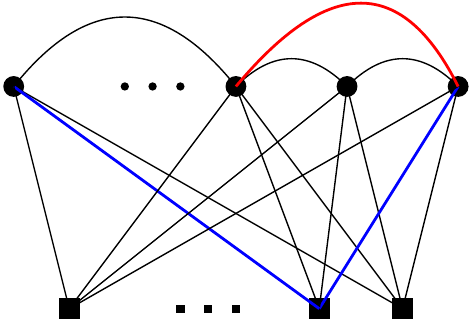}
\end{center}
\caption{The factor model corresponding to the posterior distribution. The dots represent the nodes in the adjacency graph $\mb A$, while the squares represent the variables corresponding to the Gaussian covariates $\mb B$. An $A$ edge is highlighted in red, while a $B$-wedge is indicated in blue.}
\label{fig:factor}
\end{figure}

\begin{definition}[Cycles]
For $k \geq l$, we define 
\begin{align*}
Y_{n, k, l} = \frac{1}{n^l}\sum_{\omega} \prod_{e_1\in E_{\omega,A}} A_{e_1} \prod_{e_2\in E_{\omega,B}} B_{e_2}
\end{align*}
where the sum is over length $k$ paths with $l$ $B$-wedges, and the product is over the components of each path. 
\end{definition} 
Our first result establishes the limiting distribution of these cycle statistics under the null and alternative.  To this end, it will be convenient to introduce some notation for the relevant index set in this problem. Let us define $\mathcal{J} \subset \mathbb{Z} \times \mathbb{Z}$ such that 
\begin{align}
\mathcal{J} = \{(k,0): k \geq 3\} \cup \{(k,l): k\geq l \geq 1\}. \nonumber 
\end{align}
\begin{proposition}
\label{prop:cycle_distribution} 
The collection 
\begin{align}
\{ Y_{n,k,l} : (k,l) \in \mathcal{J}\}  \nonumber 
\end{align}
converges in distribution under both $\mathrm{H}_0$ and $\mathrm{H_1}$. Further, the limiting random variables are independent under both $\mathrm{H}_0$ and $\mathrm{H}_1$. Finally, 
\begin{itemize}
\item[(i)] Under $\mathrm{H}_0$, $1\leq l \leq k$, 
\begin{align}
 Y_{n, k, 0}\stackrel{d}{\rightarrow} \mathrm{Poi}\Big(\frac{1}{k} d^k \Big), \,\,\, \frac{Y_{n, k, l} - p\mathbf{1}_{k=l=1} }{\sqrt{\frac{1}{2k}\binom{k}{l} \frac{d^{k-l}}{\gamma^l}}} \stackrel{d}{ \rightarrow} \mathcal{N}(0,1 ). \nonumber 
\end{align}
\item[(ii)] Under $H_1$, for any $1\leq l \leq k$,
\begin{align}
Y_{n, k, 0} \stackrel{d}{\rightarrow} \mathrm{Poi} \Big(\frac{1}{k}(d^k + (\lambda\sqrt{d})^k) \Big), \,\,\,\,\,
\frac{Y_{n,k,l} - p\mathbf{1}_{k=l=1} - \frac{1}{2k}\binom{k}{l}\frac{(\lambda\sqrt{d})^{k-l}\mu^l}{\gamma^l} }{\sqrt{\frac{1}{2k}\binom{k}{l} \frac{d^{k-l}}{\gamma^l}}} \stackrel{d}{\rightarrow} \mathcal{N}(0, 1). \nonumber 
\end{align} 
\end{itemize}
Finally, if $l \geq 1$, the distributional limits continue to hold for $k_n, l_n$ growing in $n$, as long as $l_n \leq k_n = o(\sqrt{\log n})$. 
\end{proposition}
 Let $\{\upsilon_{k,l,j}: j \in \{0,1\}, (k,l) \in \mathcal{J} \}$ be collection of random variables with the desired limiting distributions. Specifically, $\upsilon_{k,0,0} \sim \mathrm{Poi}\Big( \frac{1}{k} d^k \Big)$ and for $l\geq 1$, $\upsilon_{k,l,0} \sim \mathcal{N}(\mu_{k,l,0}, \sigma_{k,l}^2)$. Similarly, $\upsilon_{k,0,1} \sim \mathrm{Poi}\Big( \frac{1}{k} (d^k + (\lambda \sqrt{d})^k) \Big)$, and for $l \geq 1$, $\upsilon_{k,l,1} \sim \mathcal{N}(\mu_{k,l,1}, \sigma_{k,l}^2)$. Here, the means and variances $\mu_{k,l,0}$ and $\sigma^2_{k,l}$ are as specified in Proposition~\ref{prop:cycle_distribution} under the null, while $\mu_{k,l,1}$ denotes the mean under the alternative.

\begin{theorem}
\label{thm:contiguous_phase} 
Consider $(\lambda, \mu)$ satisfying $\lambda^2 + \frac{\mu^2}{\gamma} < 1$. Then the following hold:
\begin{enumerate}
    \item $\mathbb{P}_{\lambda, \mu}$ and $\mathbb{P}_{0,0}$ are asymptotically mutually contiguous.
    \item Under $H_0$, we have that 
    \begin{align*}
        L_n \overset{d}{\rightarrow} \exp \Big( \sum_{k = 1}^\infty \Big[ \log (1 - \lambda^kd^{k/2}) \upsilon_{k,0,0}-\frac{1}{k}(\lambda\sqrt{d})^k + \sum_{1\leq l \leq k} \frac{\mu_{k,l,0}\upsilon_{k,l,0} - \frac{1}{2}\mu_{k,l,0}^2}{\sigma_{k,l}^2} \Big]\Big).
    \end{align*}
\end{enumerate}
\end{theorem}

\subsection{Related Literature}
Covariate assisted clustering has been extensively studied across statistics, machine learning and computer science using diverse perspectives. The literature on this topic is quite diffuse, and its impossible to provide an exhaustive survey of this area. However, for the convenience of the reader, we survey the main methodological approaches, and discuss in-depth the main results relevant for our work. 

From a methodological standpoint, generative model based approaches are very natural for this problem, and they have been extensively explored in this setting \cite{newman2016structure,hoff2003random,zanghi2010clustering,yang2009combining,kim2012latent,leskovec2012learning,xu2012model,hoang2014joint,yang2013community}. On the other hand, model free approaches, which cluster the nodes by optimizing a suitable loss function have also been popular \cite{binkiewicz2017covariate,zhang2016community,gibert2012graph,zhou2009graph,neville2003clustering, gunnemann2013spectral, dang2012community,cheng2011clustering, silva2012mining, smith2016partitioning}. Bayesian methods  \cite{chang2010hierarchical,balasubramanyan2011block}  provide another natural methodological approach for this problem. We refer the interested reader to \cite{bothorel2015clustering} for a survey of other approaches.


In a separate direction \cite{aicher2014learning,lelarge2015reconstruction} study a version of community detection with informative edges. \cite{lelarge2015reconstruction} establishes only one side of the conjectured information theoretic threshold in this setting. 


More recently, \cite{binkiewicz2017covariate,zhang2016community} analyze specific heuristic clustering algorithms under the block model formalism. However, consistency guarantees in this setting are derived for dense graphs, and under strong separability assumptions on the connection probabilities. Further, they do not identify the precise information theoretic thresholds for recovery. As a consequence, the precise information theoretic gains obtained from the additional covariates remains unclear. 

Our work is closest in spirit to \cite{yan2020covariate}. They study an SDP based framework for community recovery. However, in contrast to our setting, they study low-dimensional covariates. They formally establish that clustering accuracy is improved upon combining the node information with the graph. In contrast, we study high-dimensional covariates, and establish that the information theoretic threshold is shifted in this setting. 

Somewhat related to our inquiry, \cite{kanade2016global,mossel2016local} study local algorithms for semisupervised clustering, i.e. when the true labels are given for a small fraction of nodes. While these algorithms are local, our analysis is global, and we capture the information theoretic limits in this problem.

\subsection{Technical Contributions} 

The main contributions of this paper are the algorithm for weak recovery above the threshold $\lambda^2 + \frac{\mu^2}{\gamma} > 1$, and the asymptotic distribution of the cycle statistics in Proposition~\ref{prop:cycle_distribution}. Contiguity below the threshold follows from standard second moment arguments, and expansions of the likelihood ratio in the contiguity phase are based on a version of Janson's small-subgraph conditioning method. The small-subgraph conditioning argument follows the same template as \cite{banerjee2018lr}, once the distribution of the cycle statistics Proposition~\ref{prop:cycle_distribution} is given. In this section, we elaborate on our main technical contributions.


In prior work on community detection (see e.g. \cite{mossel2015reconstruction,banerjee2018lr,hopkins2017bayesian,massoulie2014community}), weak recovery has often been performed using statistics based on appropriate self-avoiding walks. In particular, the wide-applicability of this idea was emphasized in  \cite{hopkins2017bayesian}. The general meta-algorithm introduced in \cite{hopkins2017bayesian} has the following steps: (i) estimate the second moment $\sigma \sigma^{\mathrm{T}}$ using self-avoiding walk statistics, and (ii) use a generic projection and rounding procedure to derive the membership estimator. We follow the same strategy--- the main technical challenge is to construct the appropriate estimator for the second moment. One might naturally suspect that an appropriate self-avoiding walk based statistic might be relevant in our setting. However, we have two data sources, encoded by the graph adjacency matrix and the matrix of gaussian covariates. As a consequence, the relevant cycle statistics are not obvious in this setting. A first idea is to consider the factor graph for this problem (see Figure~\ref{fig:factor}), and use self-avoiding walks of a fixed length. However, some thought reveals that this is sub-optimal; to see this, consider a setting where $\lambda^2 + \frac{\mu^2}{\gamma} >1$, but $\lambda^2 <1$. In this case, the block model alone does not contain any information regarding the underlying community assignment, so the walks based purely on the sparse graph will only contain noise.


Instead, we construct paths with edges from both the adjacency and covariate matrices, whose ratio of edges is given by $\lambda^2:\frac{\mu^2}{\gamma}$. Conceptually, each path is constructed so that the contributions from the adjacency and covariate matrices reflect the amount of information from each source respectively. This approach can be potentially useful for other reconstruction problems which have multiple information sources.

On the other hand, the distribution of the likelihood ratio in the contiguity regime $\lambda^2 + \frac{\mu^2}{\gamma} < 1$  is determined by all cycles of finite length. For cycles with edges coming solely from the adjacency or covariance matrix, the distribution limits are Poisson \cite{mossel2015reconstruction} and Gaussian \cite{banerjee2018lr} respectively. In our setting, we also encounter mixed-cycles, comprising edges coming from both sources. Using a method of moments approach, we establish that the limiting distribution of these mixed cycles are all independent Gaussian random variables in the limit. Finally, we characterize the means and variances of these cycles under the null and alternative. We expect the general techniques to be useful in other settings as well. 


\noindent
\textbf{Organization:}
The rest of the paper is structured as follows. We establish Theorem \ref{thm:it_threshold} in Section \ref{sec:detection}. In Section \ref{sec:lrtexpansion} we establish the asymptotic expansion of the likelihood ratio in the contiguity regime, and establish Theorem \ref{thm:contiguous_phase}. We establish Theorem \ref{thm:recovery} in Sections \ref{sec:ubd_below} and \ref{sec:ubd}. 

\noindent
\textbf{Acknowledgments:} 
SS thanks Yash Deshpande, Andrea Montanari and Elchanan Mossel for helpful conversations.

\section{Detection}
\label{sec:detection} 

\begin{proof}[Proof of Theorem \ref{thm:it_threshold}]
We start with a proof of the information theoretic lower bound. Fix $\lambda, \mu$ such that $\lambda^2 + \frac{\mu^2}{\gamma} <1$. We will use the traditional second moment approach. First, consider a complete data problem, where one observes the latent vectors $\sigma \in \{ \pm 1\}^n$ and $u \in \mathbb{R}^p$. We denote the corresponding distribution as $\tilde{\mathbb{P}}_{\lambda, \mu}$. Thus we have, 
\begin{align}
L:=\frac{\mathrm{d} \mathbb{P}_{\lambda, \mu} }{\mathrm{d} \mathbb{P}_{0,0}} (\mathbf{A}, \mb B)  =  \frac{ \mathbb{E}_{\sigma, u}\Big[ \tilde{ \mathbb{P}}_{\lambda, \mu} (\mathbf{A}, \mb B| \sigma, u)  \Big] }{ \mathbb{P}_{0,0} (\mathbf{A}, \mb B) }, \nonumber 
\end{align}
where $\mathbb{E}_{\sigma, u}[\cdot]$ calculates the expectation with respect to the priors on $\sigma$ and $u$. Consider the event $\mathcal{S} = \{ u : \|u\|_2 \leq  (1+\delta)\sqrt{p}\}$, where $\delta>0$ will be chosen appropriately. Define the truncated likelihood
\begin{align}
\tilde{L} = \frac{ \mathbb{E}_{\sigma, u}\Big[ \tilde{ \mathbb{P}}_{\lambda, \mu} (\mathbf{A}, \mb B |  \sigma,u)  \mathbf{1}(u \in \mathcal{S})\Big] }{ \mathbb{P}_{0,0} (\mathbf{A}, \mb B) }. \nonumber 
\end{align}
Finally, we claim that if $\lambda^2 + \frac{\mu^2}{\gamma} < 1$, then there exists a universal constant $C>0$ such that $\mathbb{E}_{0,0}[\tilde{L}^2]\leq C < \infty$.. This establishes the desired contiguity property. Indeed, let $\{ \mathcal{A}_n : n \geq 1\}$ be any sequence of events with  $\mathbb{P}_{0,0}(\mathcal{A}_n) \to 0$ as $n \to \infty$. We have, 
\begin{align}
\mathbb{P}_{\lambda, \mu} (\mathcal{A}_n) = \mathbb{E}_{0, 0}\Big[ L \mathbf{1}_{\mathcal{A}_n} \Big] = \mathbb{E}_{0,0}[ \tilde{L} \mathbf{1}_{\mathcal{A}_n} ] + \mathbb{E}_{0,0}[ (L - \tilde{L}) \mathbf{1}_{\mathcal{A}_n}]. \nonumber 
\end{align}
Note that 
\begin{align}
 \mathbb{E}_{0,0}[ (L - \tilde{L}) \mathbf{1}_{\mathcal{A}_n}]  \leq \mathbb{E}_{0,0}[(L - \tilde{L})] \leq \mathbb{P}_{\sigma,u} (u \notin \mathcal{S}) \to 0 \nonumber 
\end{align}
as $n \to \infty$. Given this claim, by Cauchy-Schwarz inequality, we have, 
\begin{align}
\mathbb{P}_{\lambda, \mu}(\mathcal{A}_n) \leq \sqrt{\mathbb{E}_{0,0}[\tilde{L}^2] \,\mathbb{P}_{0,0}(\mathcal{A}_n) } + o(1) \to 0 \nonumber 
\end{align}
as $n \to \infty$. 

It remains to establish that $\mathbb{E}_{0,0}[\tilde{L}^2] \leq C < \infty$ for some universal $C>0$. To this end, by Fubini's theorem, we note that 
\begin{align}
\mathbb{E}_{0,0}[\tilde{L}^2] = \mathbb{E}_{(\sigma, u), (\tau, v)} \Big[ \mathbb{E}_{0,0}\Big[  \frac{\tilde{\mathbb{P}}_{\lambda,\mu} (\mathbf{A}, \mb B | \sigma, u)}{\mathbb{P}_{0,0}(\mathbf{A}, \mb B)}   \frac{\tilde{\mathbb{P}}_{\lambda,\mu} (\mathbf{A}, \mb B | \tau, v)}{\mathbb{P}_{0,0}(\mathbf{A}, \mb B)} \mathbf{1}(u,v \in \mathcal{S}) \Big]\Big]. \label{eq:second_moment}
\end{align}
Now, we have, 
\begin{align}
\frac{\tilde{\mathbb{P}}_{\lambda,\mu} (\mathbf{A}, \mb B | \sigma, u)}{\mathbb{P}_{0,0}(\mathbf{A}, \mb B)}  = \frac{\tilde{\mathbb{P}}_{\lambda,\mu} (\mathbf{A}|\sigma)}{\mathbb{P}_{0,0} (\mathbf{A}) } \cdot \frac{\tilde{\mathbb{P}}_{\lambda,\mu} (\mb B|\sigma, u)}{\mathbb{P}_{0,0} (\mb B) }. \nonumber 
\end{align}
We evaluate each term in turn. First, we have, 
\begin{align}
\frac{\tilde{\mathbb{P}}_{\lambda,\mu} (\mathbf{A}|\sigma)}{\mathbb{P}_{0,0} (\mathbf{A}) } = \prod_{i<j} W_{ij}, 
W_{ij} = W_{ij}(\mathbf{A},\sigma) = \begin{cases}
\frac{2a}{a+b} &\text{ if } \sigma_i = \sigma_j, \ \ A_{ij} = 1 \\
\frac{2b}{a+b} &\text{ if } \sigma_i \not= \sigma_j, \ \ A_{ij} = 1 \\
\frac{n-a}{n-(a+b)/2} &\text{ if } \sigma_i = \sigma_j, \ \ A_{ij} = 0 \\
\frac{n-b}{n-(a+b)/2} &\text{ if } \sigma_i \not= \sigma_j, \ \ A_{ij} = 0
\end{cases} \nonumber 
\end{align}
Second, direct computation yields
\begin{align}
\frac{\tilde{\mathbb{P}}_{\lambda,\mu} (\mb B|\sigma, u)}{\mathbb{P}_{0,0} (\mb B) }
&= \exp \Big( \sqrt{\frac{\mu}{n}} \sum_{i=1}^n \sigma_i Z_i^\top u - \frac{\mu}{2}\lVert u\rVert^2 \Big).  \nonumber
\end{align}
Therefore, 
\begin{align}
& \mathbb{E}_{(\sigma, u), (\tau, v)} \Big[ \mathbb{E}_{0,0}\Big[  \frac{\tilde{\mathbb{P}}_{\lambda,\mu} (\mathbf{A}, \mb B | \sigma, u)}{\mathbb{P}_{0,0}(\mathbf{A}, \mb B)}   \frac{\tilde{\mathbb{P}}_{\lambda,\mu} (\mathbf{A}, \mb B | \tau, v)}{\mathbb{P}_{0,0}(\mathbf{A}, \mb B)} \mathbf{1}(u,v \in \mathcal{S}) \Big]\Big]  \nonumber \\
& =  \mathbb{E}_{(\sigma, u), (\tau, v)} \Big[ \mathbf{1}(u,v \in \mathcal{S})  \mathbb{E}_{0,0}\Big[ \prod_{i<j} W_{ij} V_{ij} \exp\Big( \sqrt{\frac{\mu}{n}} \sum_{i=1}^n  Z_i^\top (\sigma_i u+ \tau_i v) - \frac{\mu}{2} (\lVert u\rVert^2 +  \lVert v\rVert^2) \Big) \Big] \Big], \nonumber
\end{align}
where $V_{ij} = V_{ij}(\mathbf{A}, \tau)$ is defined similarly to $W_{ij}$. Under $\mathbb{P}_{0,0}$, for any $(\sigma, u)$ and $(\tau, v)$, $\mathbf{A}$ and $\mb B$ are independent. Setting $\rho= \rho(\sigma, \tau) = \frac{1}{n} \langle \sigma, \tau \rangle$, we have, 
\begin{align}
\mathbb{E}_{0,0}\Big[\exp\Big( \sqrt{\frac{\mu}{n}} \sum_{i=1}^n  Z_i^\top (\sigma_i u+ \tau_i v) - \frac{\mu}{2} (\lVert u\rVert^2 +  \lVert v\rVert^2) \Big) \Big] = \exp\Big( \frac{\mu}{n} \langle u,v \rangle \langle \sigma, \tau \rangle \Big). \nonumber 
\end{align}
Using \cite[Lemma 5.4]{mossel2015reconstruction}, we have, 
\begin{align}
\mathbb{E}_{0,0}\Big[\prod_{i<j} W_{ij} V_{ij} \Big] = (1 + o(1))e^{-\lambda^2/2-\lambda^4/4}\exp\Big(\frac{\rho^2\lambda^2}{2} (d + n)\Big). \nonumber 
\end{align}
Plugging these back into \eqref{eq:second_moment}, we obtain, 
\begin{align}
\mathbb{E}_{0,0}[\tilde{L}^2] 
&\leq (1 + o(1)) e^{-\lambda^2/2-\lambda^4/4 + \lambda^2 d/2}  \mathbb{E}_{(\sigma, u), (\tau, v)} \Big[ \exp \Big(n\Big( \frac{\rho^2\lambda^2}{2} + \frac{\mu}{\gamma}\rho \frac{\langle u, v \rangle}{p}\Big)\Big) \mathbf{1}( u,v \in \mathcal{S} ) \Big)  \Big].  \nonumber, 
\end{align}
We note that $\langle u, v \rangle = \|u\| \|v \| \langle \frac{u}{\|u\|}, \frac{v}{\|v\|} \rangle$, $\|u\|, \|v\| \leq (1+ \delta)\sqrt{p}$ on the event $\mathcal{S}$, and $ \langle \frac{u}{\|u\|}, \frac{v}{\|v\|} \rangle \stackrel{d}{=} Y$, where $Y$ is the first coordinate of a uniform vector on the unit sphere. Thus we have, 
\begin{align}
\mathbb{E}_{0,0}[\tilde{L}^2] \leq (1 + o(1)) e^{-\lambda^2/2-\lambda^4/4 + \lambda^2 d/2} \mathbb{E}\Big[\exp\Big( n \Big( \frac{\lambda^2}{2} X^2 + (1+\delta)^2\frac{\mu}{\gamma} XY \Big) \Big) \Big], \nonumber
\end{align} 
where $X \stackrel{d}{=} \rho(\sigma, \tau)$ and $Y$ is as described above. It is easy to see that $Y \in [-1,1]$ has density 
\begin{align}
f_Y(y) = \frac{\Gamma(p/2)}{\Gamma((p-1)/2) \Gamma(1/2)} (1-y^2)^{(p-3)/2}\leq C\sqrt{n} (1-y^2)^{p/2},  \nonumber
\end{align}
for some universal constant $C>0$. Further, for $s \in (\frac{2}{n} \mathbb{Z}) \cap [-1,1]$, 
\begin{align}
\mathbb{P}(X =s ) = \frac{1}{2^n} {n \choose n(1+s)/2} \leq \nonumber \frac{C}{\sqrt{n}} \exp(n h(s)), \nonumber 
\end{align}
where $h(s) = - (1+s)/2 \log (1+s) - (1-s)/2 \log (1-s)$. Using $h(s) \leq -s^2/2$, direct computation now yields that 
\begin{align}
\mathbb{E}\Big[\exp\Big( n \Big( \frac{\lambda^2}{2} X^2 + (1+\delta)^2\frac{\mu}{\gamma} XY \Big) \Big) \Big] \leq Cn \int_{\mathbb{R}^2} \exp\Big[ n \Big( \frac{\lambda^2}{2} s^2 + \frac{\mu}{\gamma} (1+\delta)^2 sy - \frac{s^2}{2} - \frac{y^2}{2\gamma} \Big) \Big] ds dy < C' \nonumber
\end{align}
for some universal constant $C'$, provided $\lambda^2 + \frac{\mu^2}{\gamma} (1+\delta)^2 <1$. This completes the proof. 

Next, we turn to the regime $\lambda^2 + \frac{\mu^2}{\gamma} >1$. We will devise a test based on the cycle statistic $Y_{n,k,l}$ with 
\begin{align}
\frac{l}{k} = \frac{\mu^2/\gamma}{\lambda^2 + \mu^2/\gamma}, \nonumber 
\end{align}
and some $k$, to be chosen appropriately. For $k$ growing sufficiently slowly in $n$, Proposition~\ref{prop:cycle_distribution} implies that $Y_{n,k,l}/\sigma_{k,l}$ is approximately $\mathcal{N}(0,1)$ under $\mathrm{H}_0$, while it is distributed as $\mathcal{N}(\tilde{\mu},1)$ under $\mathrm{H}_1$. The non-centrality parameter in this case is 
\begin{align}
\tilde\mu &= \frac{1}{\sqrt{2k}} \Big( \binom k l (\lambda^2)^\alpha\Big( \mu^2/\gamma \Big) ^\beta\Big)^{1/2} 
= \exp\Big(\frac 1 2 k \log(\lambda^2 + \mu^2/\gamma) + o(1) \Big). \nonumber 
\end{align} 
Thus for $k$ growing sufficiently slowly in $n$, we will get a sequence of consistent tests. This establishes the positive side of the detection threshold. 

\end{proof}

\input{lr_expansion}

\input{reconstruction}

\bibliographystyle{plain}
\bibliography{allbib}

%
%
%
%
%
%
%
%

\end{document}

%% file: lr_expansion.tex

\section{Likelihood ratio expansion} 
\label{sec:lrtexpansion} 

Armed with the distributional characterization of Proposition~\ref{prop:cycle_distribution}, we can characterlize the likelihood ratio expansion with a version of the small subgraph conditioning argument relevant to our setting. This argument was originally formalized by Robinson and Wormald \cite{robinson1992almost, robinson1994almost} in the context of random d-regular graphs, and was utilized by \cite{mossel2015reconstruction} in their study of community detection for the stochastic block model. On the other hand, inspired by a version of this argument developed by Janson \cite{janson1995random},  Banerjee and Ma \cite{banerjee2018lr} develop a Gaussian variant of this argument, and apply it to the study of contiguous regimes for Gaussian matrices with low-rank perturbations. Our setting naturally has a sparse graph, and a Gaussian component, and thus requires an extension. We expect this result to be useful in many other settings. 

\begin{proposition}[Small subgraph conditioning method] 

Let $\mathbb{P}_n$ and $\mathbb{Q}_n$ be two sequences of probability measures, and let  $\{ Y_{n, k, l}: (k,l) \in \mathcal{J}\}$ be such that the following conditions hold:
\begin{enumerate}
    \item $\mathbb{Q}_n$ is absolutely continuous w.r.t. $\mathbb{P}_n$.
    \item All finite dimensional distributions of $\{ Y_{n, k, l}: (k,l) \in \mathcal{J}\}$ converge to the null distribution under $\mathbb{P}_n$, and to the alternative distribution under $\mathbb{Q}_n$, as specified in Proposition~\ref{prop:cycle_distribution}.
    \item The likelihood ratio $L_n = \frac{\mathrm{d}\mathbb{Q}_n}{\mathrm{d} \mathbb{P}_n}$ satisfies:
    \begin{align}\label{eq:bdd_2nd_mmt}
        \limsup_{n\rightarrow \infty} \mathbb{E}_{\mathbb{P}_n}{[L_n^2]} \leq \exp\Big\{-\frac{1}{2}\log(1 - (\lambda^2 + \frac{\mu^2}{\gamma})) - \frac{\lambda^2}{2} - \frac{\lambda^4}{4}\Big\} < \infty. \nonumber 
    \end{align}
\end{enumerate}
Then, we have the following consequences:
\begin{enumerate}
    \item $\mathbb{P}_n$ and $\mathbb{Q}_n$ are asymptotically mutually contiguous.
    \item Under $\mathbb{P}_n$, we have that
    \begin{align*}
        L_n \overset{d}{\rightarrow} \exp \Big(\sum_{k = 1}^\infty \Big[\log (1 - \lambda^kd^{k/2}) \upsilon_{k,0,0}-\frac{1}{k}(\lambda\sqrt{d})^k + \sum_{1\leq l \leq k} \frac{\mu_{k,l,0}\upsilon_{k,l,0} - \frac{1}{2}\mu_{k,l,0}^2}{\sigma_{k,l}^2} \Big]\Big).
    \end{align*}
\end{enumerate}
\end{proposition} 
Given Proposition~\ref{prop:cycle_distribution}, the proof of this proposition is identical to the proof of Proposition 1 of \cite{banerjee2018lr}, but with some Gaussian terms swapped out with Poisson terms. Thus we omit the proof.

\subsection{Proof of Proposition~\ref{prop:cycle_distribution}}
We prove Proposition~\ref{prop:cycle_distribution} in this section. Formally, fix $(k_1, l_1), \cdots, (k_r, l_r) \in \mathcal{J}$, and $m_1, \cdots, m_r \geq 1$.  Without loss of generality, we assume that there exists $r_1 \leq r$ such that $l_1 = \cdots = l_{r_1} = 0$ and $l_{j} >0$ for $j > r_1$. Further, assume that $k_1 < k_2 < \cdots < k_{r_1}$ and $k_{r_1+1}< k_{r_1+2}< \cdots < k_r$. For convenience, we will denote
\begin{align}
&Z_{n, k_j,0} = Y_{n,k_j,0}, \,\,\,\, 1\leq j \leq r_1, \nonumber\\
& Z_{n,k_{r_1+1}, l_{r_1+1}}= Y_{n,k_{r_1+1},l_{r_1+1}} - p \mathbf{1}_{k_{r_1+1} = l_{r_1+1}=1} , Z_{k_j,l_j} = Y_{n,k_j,l_j}  ,\,\,\,\,\, j > r_1. \nonumber
\end{align}
We will show that for $j\leq r_1$, $Z_{n, k_j, l_j}$ have Poisson limits, and for $j > r_1$, $Z_{n, k_j, l_j}$ have Gaussian limits. This is done with the method of moments: we will establish that as $n \to \infty$, 
\begin{equation}
\begin{aligned}\label{eq:moment_lim}
\mathbb{E}_{0,0}\Big[\prod_{j=1}^{r} Z_{n, k_j, l_j}^{m_j} \Big] &\to \prod_{j=1}^{r_1} \mathbb{E}[\upsilon_{k_j,l_j,0}^{m_j}] \cdot \prod_{j=r_1+1}^{r} \sigma_{k_j,l_j}\mathbb{E}[\xi_j^{m_j}],\\ 
\mathbb{E}_{\lambda,\mu}\Big[\prod_{j=1}^{r} Z_{n, k_j, l_j}^{m_j} \Big] &\to \prod_{j=1}^{r_1} \mathbb{E}[\upsilon_{k_j,l_j,1}^{m_j}] \cdot \prod_{j=r_1+1}^{r} \sigma_{k_j,l_j} \mathbb{E}[\xi_j^{m_j}], 
\end{aligned}
\end{equation} 
where $\{\xi_1, \xi_2, \cdots\}$ is a sequence of iid $\mathcal{N}(0,1)$ random variables. It is easy to verify that the random variables $Z_{n, k_j, l_j}$ satisfy Carleman's condition, so \eqref{eq:moment_lim} implies the desired convergence in distribution in Proposition~\ref{prop:cycle_distribution}.  

We first establish a \emph{decoupling} lemma, which will allow us to separate the analysis of terms with $l=0$, which have Poisson limits, from terms with $l>0$, which have Gaussian limits. 

\begin{lemma}
\label{lemma:decoupling} 
Fix $r \geq 1$ and $r_1 \leq r$. Fix $(k_1,l_1), \cdots, (k_r,l_r) \in \mathcal{J}$, and $m_1,\cdots, m_r \geq 1$. Further assume that $l_j =0$ for $j \leq r_1$, and $l_j>0$ for $r_1<j\leq r$. Then we have, as $n \to \infty$, 
\begin{align}
\Big| \mathbb{E}_{0,0}\Big[ \prod_{j=1}^{r} Z_{n, k_j, l_j}^{m_j}  \Big] - \mathbb{E}_{0,0}\Big[ \prod_{j=1}^{r_1} Z_{n,k_j, l_j}^{m_j} \Big] \mathbb{E}_{0,0}\Big[ \prod_{j=r_1 +1}^{r} Z_{n,k_j,l_j}^{m_j} \Big] \Big| \to 0, \nonumber \\
\Big| \mathbb{E}_{\lambda,\mu}\Big[ \prod_{j=1}^{r} Z_{n, k_j, l_j}^{m_j}  \Big] - \mathbb{E}_{\lambda,\mu}\Big[ \prod_{j=1}^{r_1} Z_{n,k_j, l_j}^{m_j} \Big] \mathbb{E}_{\lambda,\mu}\Big[ \prod_{j=r_1 +1}^{r} Z_{n,k_j,l_j}^{m_j} \Big] \Big| \to 0. \nonumber 
\end{align}
\end{lemma}

\begin{proof}
We first establish the assertion under $\mathrm{H}_0$. We have 
\begin{align}
&\mathbb{E}_{0,0}\Big[ \prod_{j=1}^{r} Z_{n,k_j,l_j}^{m_j} \Big] = \mathbb{E}_{0,0} [T_1 \cdot T_2 \cdot T_3] , \nonumber \\
T_1 &=  \prod_{j=1}^{r_1} \Big( \sum_{\omega_j}  \prod_{e \in E_{\omega_j,A}} A_e  \Big)^{m_j},  \nonumber \\
 T_2 &= \Big( \frac{1}{n^{l_{r_1+1}}} \sum_{\omega_{r_1+1}} \prod_{e  \in E_{\omega_{r_1+1} ,A}} A_e  \prod_{e  \in E_{\omega_{r_1+1} ,B}} B_e -  p \cdot \mathbf{1}(k_{r_1+1} = l_{r_1+1} =1) \Big)^{m_{r_1 +1}} ,\nonumber \\
 T_3 &= \prod_{j=r_1 +2}^{r}  \Big( \frac{1}{n^{l_j}} \sum_{\omega_j} \prod_{e \in E_{\omega_j,A}} A_e  \prod_{e \in E_{\omega_j,B}} B_e\Big)^{m_j}. \nonumber 
\end{align}
Expanding, these terms may be expressed as 
\begin{align}
T_1 &= \prod_{j=1}^{r_1} \Big( \sum_{\omega_{j,1}, \cdots , \omega_{j,m_j}} \prod_{q=1}^{m_j} \prod_{e \in E_{\omega_{j,q},A}} A_e \Big),  \nonumber \\
T_2 &= \frac{1}{n^{l_{r_1+1} m_{r_1+1}}} \sum_{\omega_{r_1+1,1}, \cdots, \omega_{r_1+1,m_{r_1+1}}} \prod_{q=1}^{m_{r_1+1}} \Big( \prod_{e \in E_{\omega_{r_1+1,q},A} } A_e \prod_{e  \in E_{\omega_{r_1+1,q} ,B}} B_e - p \cdot \mathbf{1}(k_{r_1+1} = l_{r_1+1} =1) \Big), \nonumber \\
T_3 &= \prod_{j=r_1+2}^{r} \Big(  \frac{1}{n^{l_j m_j}} \sum_{\omega_{j,1}, \cdots, \omega_{j,m_j} } \prod_{q=1}^{m_j} \prod_{e \in E_{\omega_{j,q},A}} A_e \prod_{e \in E_{\omega_{j,q},B}} B_{e}   \Big). \nonumber 
\end{align}
Thus we have, 
\begin{align}
&\mathbb{E}_{0,0}[T_1 T_2 T_3| \mathbf{A}] \nonumber \\
&= \frac{1}{n^{\sum_{j= r_1+1}^{r} l_j m_j }} \sum_{j \in [r], 1\leq q_j \leq m_j, \omega_{j,q_j}} \mathbb{E}_{0,0}
\left[\begin{matrix}
 \Big( \prod_{j=1}^{r_1} \prod_{q_j \leq m_j}  \prod_{e \in E_{\omega_{j,q_j},A} }A_e  \Big) \times \\
 \prod_{q=1}^{m_{r_1+1}} \Big( \prod_{e \in E_{\omega_{r_1+1,q},A} } A_e \prod_{e  \in E_{\omega_{r_1+1,q} ,B}} B_e - p \cdot \mathbf{1}(k_{r_1+1} = l_{r_1+1} =1) \Big) \times  \\
 \prod_{j=r_1+2}^{r} \prod_{q_j=1}^{m_j} \prod_{e \in E_{\omega_{j,q_j},A}} A_e \prod_{e \in E_{\omega_{j,q_j},B}} B_{e}
\end{matrix} \Big| A \right]. \nonumber 
\end{align}
Consider first the case $(k_{r_1+1}, l_{r_1+1}) \neq (1,1)$. In this case, 
\begin{align}
&\mathbb{E}_{0,0}[T_1 T_2 T_3] = \frac{1}{n^{\sum_{j= r_1+1}^{r} l_j m_j }} \times \\
& \sum_{j \in [r], 1\leq q_j \leq m_j, \omega_{j,q_j}}  \mathbb{E}_{0,0}\Big[\Big( \prod_{j=1}^{r_1} \prod_{q_j \leq m_j}  \prod_{e \in E_{\omega_{j,q_j},A} }A_e  \Big) 
\Big( \prod_{j=r_1+1}^{r} \prod_{q_j\leq m_j} \prod_{e \in E_{\omega_{j,q_j},A}} A_e \Big) \Big]  \mathbb{E}_{0,0}\Big[\prod_{j=r_1+1}^{r} \prod_{q_j\leq m_j} \prod_{e \in E_{\omega_{j,q_j},B}} B_{e} \Big]. \nonumber
\end{align} 
On the other hand, a similar calculation yields
\begin{align}\label{eq:lemma2}
&\mathbb{E}_{0,0}\Big[ \prod_{j=1}^{r_1} Z_{n,k_j, l_j}^{m_j} \Big] \mathbb{E}_{0,0}\Big[ \prod_{j=r_1 +1}^{r} Z_{n,k_j,l_j}^{m_j} \Big] =  \frac{1}{n^{\sum_{j= r_1+1}^{r} l_j m_j }} \times \nonumber \\
&\sum_{j \in [r], 1\leq q_j \leq m_j, \omega_{j,q_j}}  \mathbb{E}_{0,0}\Big[\Big( \prod_{j=1}^{r_1} \prod_{q_j \leq m_j}  \prod_{e \in E_{\omega_{j,q_j},A} }A_e  \Big)  \Big] \mathbb{E}_{0,0}\Big[
\Big( \prod_{j=r_1+1}^{r} \prod_{q_j\leq m_j} \prod_{e \in E_{\omega_{j,q_j},A}} A_e \Big) \Big]  \times \\
&\ \ \ \ \ \ \ \ \ \ \ \ \ \ \ \ \mathbb{E}_{0,0}\Big[\prod_{j=r_1+1}^{r} \prod_{q_j\leq m_j} \prod_{e \in E_{\omega_{j,q_j},B}} B_{e} \Big]. \nonumber
\end{align}
To establish that the difference between the two quantities is $o(1)$, note that if two cycles overlap on $m$-edges, we gain a factor $O(n^m)$. However, this naturally implies they overlap on at least $(m+1)$ vertices. As a result, we lose a factor $n^{m+1}$ in choosing these cycles. As a result, the dominant contributions arise from non-overlapping cycles, thus establishing the desired claim in this case. The proof for $(k_{r_1+1}, l_{r_1+1}) = (1,1)$ is exactly analogous. 

Under $\mathbb{P}_{\lambda,\mu}$, if $(k_{r_1+1}, l_{r_1+1} ) \neq (1,1)$, we have, 
\begin{align}
&\mathbb{E}_{\lambda,\mu}\Big[ \prod_{j=1}^{r} Z_{n,k_j,l_j}^{m_j} \Big] = \frac{1}{n^{\sum_{j= r_1+1}^{r} l_j m_j }} \times \nonumber \\ 
&  \mathbb{E}_{\sigma}\Big[ \sum_{j \in [r], 1\leq q_j \leq m_j, \omega_{j,q_j}} \mathbb{E}_{\lambda ,\mu}\Big[\Big( \prod_{j=1}^{r_1} \prod_{q_j \leq m_j}  \prod_{e \in E_{\omega_{j,q_j},A} }A_e  \Big) 
\Big( \prod_{j=r_1+1}^{r} \prod_{q_j\leq m_j} \prod_{e \in E_{\omega_{j,q_j},A}} A_e \Big) \mid \sigma \Big]  \times \\
&\ \ \ \ \ \ \ \ \ \ \ \ \ \ \ \ \mathbb{E}_{\lambda,\mu}\Big[\prod_{j=r_1+1}^{r} \prod_{q_j\leq m_j}\prod_{e \in E_{\omega_{j,q_j},B}} B_{e} \mid \sigma \Big] \Big]. \nonumber
\end{align}
On the other hand, 
\begin{align}
&\mathbb{E}_{\lambda,\mu}\Big[ \prod_{j=1}^{r_1} Z_{n,k_j, l_j}^{m_j} \Big] \mathbb{E}_{\lambda,\mu}\Big[ \prod_{j=r_1 +1}^{r} (Z_{n,k_j,l_j})^{m_j} \Big]  = \frac{1}{n^{\sum_{j= r_1+1}^{r} l_j m_j }} \times \nonumber \\  
& \mathbb{E}_{\sigma}\Big[ \sum_{j \in [r], 1\leq q_j \leq m_j, \omega_{j,q_j}}  \mathbb{E}_{\lambda,\mu}\Big[\Big( \prod_{j=1}^{r_1} \prod_{q_j \leq m_j}  \prod_{e \in E_{\omega_{j,q_j},A} }A_e  \Big) \mid \sigma \Big] \mathbb{E}_{\lambda,\mu}\Big[
\Big( \prod_{j=r_1+1}^{r} \prod_{q_j\leq m_j} \prod_{e \in E_{\omega_{j,q_j},A}} A_e \Big)  \mid \sigma \Big]  \times \\
&\ \ \ \ \ \ \ \ \ \ \ \ \ \ \ \ \mathbb{E}_{\lambda,\mu}\Big[\prod_{j=r_1+1}^{r} \prod_{q_j\leq m_j}\prod_{e \in E_{\omega_{j,q_j},B}} B_{e} \mid \sigma \Big] \Big]. \nonumber
\end{align}
To see that the difference between the quantities is again $o(1)$, first condition on the choice of $\sigma$, and consider the difference in the inner sums. Again, if two cycles overlap on $m$ edges, we will gain a factor of $O(n^m)$, but we lose a factor of $n^{m+1}$ for the number of such choices. Thus the dominant contribution again comes from the non-overlapping cycles. The proof from $(k_{r_1+1}, l_{r_1 + 1}) = (1,1)$ follows analogously.
\end{proof} 

\noindent 
Armed with Lemma \ref{lemma:decoupling}, we turn to a proof of Proposition~\ref{prop:cycle_distribution}. 
\begin{proof}[Proof of Propostion~\ref{prop:cycle_distribution}]
First we note that Lemma~\ref{lemma:decoupling} immediately implies that the terms $Z_{n, k_i, l_i}$ with $i\leq r_1$ are asymptotically independent to the terms $Z_{n, k_j, l_j}$, with $j>r_1$. Moreover, the limiting distributions of $\{Y_{n,k,0}: k \geq 3\}$ under both $\mathrm{H}_0$ and $\mathrm{H_1}$ follows directly from \cite[Theorem 3.1]{mossel2015reconstruction}. 
Thus, we have, as $n \to \infty$, 
\begin{align}
\mathbb{E}_{0,0}\Big[ \prod_{j=1}^{r_1} Z_{n,k_j, 0}^{m_j} \Big] \to \prod_{j=1}^{r_1} \mathbb{E}[\upsilon_{k_j,0,0}^{m_j}], \,\,\,\,\, \mathbb{E}_{\lambda,\mu}\Big[ \prod_{j=1}^{r_1} Z_{n,k_j, 0}^{m_j} \Big] \to \prod_{j=1}^{r_1} \mathbb{E}[\upsilon_{k_j,0,1}^{m_j}]. \nonumber 
\end{align} 
Thus it suffices to analyze the terms with $l>0$, and we show that in this case $Z_{n, k, l}$ will have asymptotically Gaussian limits under $\mrm H_0$ and $\mrm H_1$. 

\noindent 
\textbf{Calculation under $\mathrm{H}_0$:} We establish the limiting distribution in three steps. 
\begin{itemize}
\item[(i)] Calculation of the mean and variance of $Y_{n,k,l} - p \mathbf{1}(k=l=1)$ under the null. 
\item[(ii)] Verification of Wick's formula. 
\item[(iii)] Verification of asymptotic independence. 
\end{itemize}
We address (i), and calculate the means and variances of the cycle statistics. Note that the case $l=k$ corresponds to the cycle statistics in \cite{banerjee2018lr}, and we can read off the null expectations and variances directly. Specifically, we have, 
\begin{align}
&\mathbb{E}_{0,0}[Y_{n,k,k} - p \mathbf{1}(k=1)] = 0, \nonumber \\
&\mathrm{Var}_{0,0}\Big( Y_{n,k,k} - \mathbf{1}(k=1) \Big) = (1+o(1))2k\gamma^k. \nonumber
\end{align}
We consider now the case $0<l<k$. We have 
\begin{align}
\mathbb{E}_{0,0}[Y_{n,k,l}]=0. \nonumber
\end{align}
Moving onto the variance, we have, 
\begin{align}
\mathbb{E}_{0,0}[Y_{n,k,l}^2] &= \frac{1}{n^{2l}} \sum_{\omega_1, \omega_2} \mathbb{E}_{0,0}\Big[ \Big( \prod_{e_1 \in E_{\omega_1, A} } A_{e_1} \prod_{e_2 \in E_{\omega_1,B}} B_{e_2}  \Big)  \Big( \prod_{e_1 \in E_{\omega_2, A} } A_{e_1} \prod_{e_2 \in E_{\omega_2,B}} B_{e_2} \Big)\Big] \\
&= \frac{1}{n^{2l}}  \sum_{\omega}  \mathbb{E}_{0,0}\Big[ \Big( \prod_{e_1 \in E_{\omega, A} } A_{e_1} \prod_{e_2 \in E_{\omega,B}} B_{e_2}  \Big) ^2 \Big] + T_1, 
\nonumber
\end{align}
where $T_1$ tracks the contribution from the pairs $(\omega_1, \omega_2)$ with $\omega_1 \neq \omega_2$. First, observe that 
\begin{align}
\frac{1}{n^{2l}}  \sum_{\omega}  \mathbb{E}_{0,0}\Big[ \Big( \prod_{e_1 \in E_{\omega, A} } A_{e_1} \prod_{e_2 \in E_{\omega,B}} B_{e_2}  \Big) ^2 \Big] = \frac{1}{n^{2l}} \sum_{\omega} \mathbb{E}_{0,0}\Big[ \prod_{e_1 \in E_{\omega, A} } A_{e_1}  \prod_{e_2 \in E_{\omega,B}} B_{e_2}^2\Big] = \frac{1}{n^{2l}} \sum_{\omega} \Big( \frac{d}{n} \Big)^{k-l}. \nonumber
\end{align}
It remains to count the number of length $k$ cycles with $l$ $B$-wedges. Any such cycle has $k$-vertices of type $A$, and this choice can be done in $n^k$ ways. The positions of the $B$-wedges can be chosen in ${k \choose l}$ ways. The $B$-vertices on the $B$-wedges can be chosen in $p^{l}$ ways. Finally, we divide this count by $2k$ to account for overcounting due to cyclic shifts. This implies 
\begin{align}
&\frac{1}{n^{2l}}  \sum_{\omega}  \mathbb{E}_{0,0}\Big[ \Big( \prod_{e_1 \in E_{\omega, A} } A_{e_1} \prod_{e_2 \in E_{\omega,B}} B_{e_2}  \Big) ^2 \Big] = \frac{1}{n^{2l}} \sum_{\omega} \Big( \frac{d}{n} \Big)^{k-l} = \frac{1}{n^{2l}} \cdot \Big( \frac{d}{n} \Big)^{k-l}   \cdot \frac{1}{2k} n^k {k \choose l} p^l \nonumber \\
&= \frac{1}{2k} {k \choose l} \frac{d^{k-l}}{\gamma^l} (1+o(1)). \nonumber 
\end{align}
We will next establish that this is the dominant term in the asymptotic variance, and that $T_1 \to 0$ as $n \to \infty$. Note that a product term corresponding to $(\omega_1, \omega_2)$ with $\omega_1 \neq \omega_2$ has a non-zero contribution provided they share exactly the same $B$-wedges. For any two such cycles $(\omega_1, \omega_2)$, suppose they share $\alpha_1$ $A$ edges. Note that $\omega_1 \neq \omega_2$, and thus $0\leq \alpha_1 < k-l$. As $\omega_1$ and $\omega_2$ cannot differ on exactly one edge, in fact, this implies $\alpha_1 \leq k-l-2$. Setting $\alpha_2 = k-l-\alpha_1$, we have, 
\begin{align}
T_1 &= \frac{1}{n^{2l}} \sum_{\omega_1 \neq \omega_2} \mathbb{E}_{0,0}\Big[ \Big( \prod_{e_1 \in E_{\omega_1, A} } A_{e_1} \prod_{e_2 \in E_{\omega_1,B}} B_{e_2}  \Big)  \Big( \prod_{e_1 \in E_{\omega_2, A} } A_{e_1} \prod_{e_2 \in E_{\omega_2,B}} B_{e_2} \Big)\Big] \nonumber \\
&= \frac{1}{n^{2l}} \sum_{\alpha_2 = 2}^{k-l} \sum_{|E_{\omega_1,A} \cap E_{\omega_2,A}| = k-l - \alpha_2} \Big( \frac{d}{n} \Big)^{k-l+\alpha_2}. \nonumber 
\end{align}
It remains to count the number of pairs $(\omega_1, \omega_2)$ with $|E_{\omega_1, A} \cap E_{\omega_2,A}| = k-l- \alpha_2$ for all $2\leq \alpha_2 \leq k-l$. We derive a rough upper bound to the number of such pairs as follows: there are $O(n^k p^l)$ choices for the first cycle, and $O(n^{\alpha_2-1})$ choices for the second cycle $\omega_2$, given $\omega_1$. Thus we bound the number of such pairs as $C(k, \gamma) n^{k+l + \alpha_2 -1}$, where $C(k,\gamma)>0$ is independent of $n$. Plugging in this bound, we obtain 
\begin{align}
T_1 \leq C(k,\gamma)  \frac{1}{n^{2l}} \sum_{\alpha_2 =2}^{k-l} n^{k+l +\alpha_2 -1} \cdot \Big(\frac{d}{n} \Big)^{k-l + \alpha_2} = O\Big(\frac{1}{n} \Big). \nonumber 
\end{align}
This controls $T_1$, and establishes the right order of the variance under $\mathrm{H}_0$.

We next turn to (ii). We want to show that $Z_{n, k, l}$ are asymptotically Gaussian by showing the limit of their moments \eqref{eq:moment_lim}. This is done by checking that the limits of the moments satisfy Wick's formula. Formally, we show that for $W_{ni} \in \{ Z_{n, k_{r_1+1}, l_{r_1+1}}, \cdots , Z_{n,k_r,l_r} \}$, $i\in [m]$, we have that



\begin{equation}\label{eq:wick}
\lim_{n\to\infty}\mathbb{E}[W_{n1} \cdots W_{nm}] =
\begin{cases}
\sum_{\eta} \prod_{i=1}^{m/2} \mathbb{E}[W_{n\eta(i,1)} W_{n\eta(i,2)}] + o(1) &\textrm{if}\,\, m \,\, \textrm{even} \\
o(1) & \textrm{o.w.}
\end{cases}
\end{equation}
where $\eta$ is a partition of $[m]$ into $\frac m 2$ blocks of size two, and $\eta(i,j)$ denotes the $j$-th element of the $i$-th block, where $j\in\{1, 2\}$. Wick's formula \cite{wick1950evaluation} then implies that the limiting distribution must be Gaussian, as long as the limits of $\mathbb{E}[W_{n\eta(i,1)} W_{n\eta(i,2)}]$ exist.

We will perform the calculations assuming that $(k_{r_1+1},l_{r_1+1}) \not=(1,1)$; the same calculations hold in the case of equality and hence are omitted. For a choice of $W_{n1}, \cdots, W_{nm}$, let $\omega_{1:m}$ be a collection of cycles $\omega_1, ..., \omega_m$, such that $\omega_i$ is of length $k_i$, with $l_i$ $B$ wedges and $x_i$ contiguous blocks of $A$ type edges and $B$ type wedges. Note then that
\begin{align}\label{eq:gauss_expand}
\begin{aligned}
    &\mathbb{E}_{0,0}[W_{n1} \cdots W_{nm}]\\
    &= \mathbb{E}_{0,0}\Big[{n^{- \sum_{i}l_i} \sum_{\omega_{1:m}} \prod_{i\leq m} \prod_{e_1\in E_{\omega_i,A}}A_{e_1}\prod_{e_2\in E_{\omega_i,B}}B_{e_2}} \Big]\\
    &= n^{-\sum_{i}l_i} \sum_{\omega_{1:m}} \mathbb{E}_{0,0}\Big[{\prod_{i\leq m} \prod_{e_1\in E_{\omega_i,A}}A_{e_1}} \Big] \mathbb{E}_{0,0}\Big[{\prod_{i\leq m}\prod_{e_2\in \tilde E_{\omega_i}}B_{e_2}}\Big]
\end{aligned}
\end{align}
where $\tilde E_{\omega_i}$ are the edges of $\tilde G_{\omega_i} = G_{\omega_i} / G_{\omega_i, A}$, the quotient graph where for each $j\leq x_i$, $G_{\omega_i, \alpha_j}$, the graph of the $j$th $A$ block in $\omega_i$, is identified as a vertex of $\tilde G_{\omega_i}$. We denote the vertices of $\tilde G_{\omega_i}$ as $\tilde V_{\omega_i} = \tilde V^1_{\omega_i} \cup \tilde V^2_{\omega_i}$, where $\tilde V^1$ are the vertices inherited from $G_{\omega_i}$, and $\tilde V^2$ are the vertices produced by the quotient operator. Among $\tilde V^2$, we define the following equivalence relationship: if the first and last vertices of $G_{\omega_i, \alpha^i_j}$ and $G_{\omega_h,\alpha^h_p}$ are the same, then we consider the vertices in $\tilde G_{\omega_i}$ and $\tilde G_{\omega_h}$, which correspond to the quotient image of $G_{\omega_i, \alpha^i_j}$ and $G_{\omega_h,\alpha^h_p}$ respectively, to be the same. In order for the contribution of $\omega_{1:m}$ to be non-zero, the $B$ edges have to be included at least twice, and we will call such a collection $\tilde G_{\omega_{1:m}}$ a weak CLT sentence. Given a weak CLT sentence, we define a partition $\eta(\tilde G_{\omega_{1:m}})$ of $[m]$ as follows: $i$ and $j$ are in the same partition if $\tilde G_{\omega_i}$ and $\tilde G_{\omega_j}$ share at least one edge. As a result, we can express \eqref{eq:gauss_expand} as follows:
\begin{align*}
    n^{-\sum_{i}l_i} \sum_{G_{\omega_{1:m},A}} \E_{0,0}\Big[{ \prod_{e_1\in E_{\omega_{1:m},A}}A_{e_1}}\Big ] \sum_\eta \sum_{\substack{\tilde G_{\omega_{1:m}} \\ \eta(\tilde G_{\omega_{1:m}}) = \eta}} \E_{0,0}\Big[{\prod_{e_2\in \tilde E_{\omega_{1:m}}}B_{e_2}}\Big].
\end{align*}
Let $t$ be the total number of vertices of $\tilde G_{\omega_{1:m}}$. Let us consider the case where $\eta(\tilde G_{\omega_{1:m}})$ contains strictly less than $\frac m 2$ blocks, which includes all cases when $m$ is odd. In this case \cite[Lemma 4.10]{anderson2006clt} implies that $t < \sum_i l_i$. Following the proof of \cite[Lemma 3]{banerjee2018lr}, we note that the number of weak CLT sentences summed over is bounded by
\begin{align*}
    O\Big(\sum_i l_i\Big)^{O(\sum_i l_i)} n^{t - \sum_i x_i}.
\end{align*}
This comes from the fact that $\sum_i x_i$ of the vertices are automatically fixed from the quotient operator. As a result, for a particular partition $\eta$, where $|\eta| < \frac m 2$, we note that:
\begin{align*}
    & n^{- \sum_{i}l_i} \sum_{G_{\omega_{1:m},A}} \mathbb{E}_{0,0}\Big[{ \prod_{e_1\in E_{\omega_{1:m},A}}A_{e_1}} \Big] \sum_{\substack{\tilde G_{\omega_{1:m}} \\ \eta(\tilde G_{\omega_{1:m}} = \eta)}} \mathbb{E}_{0,0} \Big[{\prod_{e_2\in \tilde E_{\omega_{1:m}}}B_{e_2}} \Big]\\
    &= n^{- \sum_{i}l_i} \sum_{G_{\omega_{1:m},A}} \mathbb{E}_{0,0}\Big[{ \prod_{e_1\in E_{\omega_{1:m},A}}A_{e_1}} \Big] \sum_{\substack{\tilde G_{\omega_{1:m}} \\ \eta(\tilde G_{\omega_{1:m}} = \eta)}} O(1)^{O(\sum_i l_i)}\\
    &= n^{- \sum_{i}l_i} O\Big(\sum_i l_i\Big)^{O(\sum_i l_i)} n^{t - \sum_i x_i} \sum_{G_{\omega_{1:m},A}} \mathbb{E}_{0,0}\Big[{ \prod_{e_1\in E_{\omega_{1:m},A}}A_{e_1}} \Big]\\
    &= n^{- \sum_{i}l_i} O\Big(\sum_i l_i\Big)^{O(\sum_i l_i)} n^{t - \sum_i x_i} O\Big(\frac d n\Big)^{\sum_i \sum_{j\leq x_i}2\alpha^i_j} O(n)^{\sum_i \sum_{j\leq x_i} (2\alpha^i_j + 1)}\\
    &= O\Big(\sum_i l_i\Big)^{O(\sum_i l_i)} O(n)^{t - \sum_i l_i}.
\end{align*}
The penultimate line holds because for an $A$ block of length $2\alpha^i_j$, there are $2\alpha^i_j + 1$ vertices, for which there are $O(n)$ options each. Thus the total number of choices for $E_{\omega_{1:m},A}$ is $O(n)^{\sum_i \sum_{j\leq x_i}(2\alpha^i_j + 1)}$. Since $t< \sum_i l_i$, we see that the contribution of such a term is $o(1)$.
%

We have thus shown that the leading order term of \eqref{eq:gauss_expand} consists of weak CLT sentences whose partition $\eta(\tilde G_{\omega_{1:m}})$ has exactly $\frac m 2$ blocks. Note that this automatically implies that \eqref{eq:gauss_expand} is $o(1)$ if $m$ is odd, and that for $m$ even, the leading order weak CLT sentences have partitions $\eta$ with only blocks of size two, which is exactly what we need for \eqref{eq:wick}. From our calculations in step (i), we see that the variances of $Z_{n, k, l}$ under the null are as we claimed.

Finally, we verify step (iii), that for $l\geq 1$, $Z_{n, k,l}$ are asymptotically independent. From the discussion above \eqref{eq:wick}, this amounts to checking that for $W_{ni}\not=W_{nj}$, $\E_{0,0}[W_{ni}W_{nj}] \overset{p}{\to} 0$. Note that this expectation equals $0$ if $l_i\not= l_j$. Thus it suffices to consider the case $l_i=l_j$, $k_i< k_j$ and we have:
\begin{align*}
    \E_{0,0}[W_{ni}W_{nj}] &= n^{-2l_i} \sum_{\omega_i, \omega_j}\EE[0,0]{\prod_{e_1\in E_{\omega_i}} A_{e_1} \prod_{e_1\in E_{\omega_j}} A_{e_1}}\\
    \intertext{where the sum is taken over $\omega_i$, $\omega_j$ that intersect on all $l_i$ $B$-wedges, which gives}\\
    \E_{0,0}[W_{ni}W_{nj}] &\leq  C(k_i, d, \gamma) n^{-2l_i} \p{\frac 1 n}^{k_i + k_j - 2l_i} n^{k_i + k_j - l_i - 1} = o(1).
\end{align*}
This concludes the demonstration of Wick's formula under $H_0$, and hence proves the desired convergence in distribution.

\noindent 
\textbf{Calculation under $\mathrm{H}_1$:} We establish the desired result following similar steps as the calculation under $\mathrm{H}_0$. First, we calculate the means of the cycle statistics under $\mathrm{H}_1$. 

We start by observing that the case $k=l=1$ follows directly from the calculations in \cite{banerjee2018lr}, while the $l=0$ cases correspond to $A$ type cycles, and have been worked out in \cite{mossel2015reconstruction}. We consider the remaining cases. To this end, note that under $\mathbb{P}_{\lambda, \mu}$, for any $B$-wedge $B_{i_1,j_1} B_{i_2,j_1}$, we have,  
\begin{align}
\mathbb{E}_{\lambda, \mu}[B_{i_1,j_1} B_{i_2,j_1} | \mathbb{\sigma}] = \frac{\mu}{n} \sigma_{i_1} \sigma_{i_2}. \nonumber
\end{align}
Now fix any cycle $\omega$ with $k-l$ A-edges and $l$ B-wedges. 
\begin{align}
&\mathbb{E}_{\lambda,\mu}[Y_{n, k, l} ] = \frac{1}{n^l}\sum_{\omega} \mathbb{E}_{\lambda,\mu}\Big[  \prod_{e_1\in E_{\omega,A}} A_{e_1} \prod_{e_2\in E_{\omega,B}} B_{e_2} \Big] \nonumber \\
&= \frac{1}{n^l} \sum_{\omega} \mathbb{E}_{\sigma} \Big[ \prod_{e_1 \in E_{\omega,A}} \Big( \frac{d + \lambda \sqrt{d} \, \sigma_{e_1^-} \sigma_{e_1^+} }{n} \Big) \prod_{e_2 \in E_{\omega,B}} \Big(\frac{\mu}{n} \sigma_{e_2^-} \sigma_{e_2^+} \Big)  \Big].  \nonumber 
\end{align}
In the final equality, for an $A$-type edge $e_1$, we use $e_1^- , e_1^+$ to denote its end points. Similarly, for a $B$-wedge $\{i,j,k\}$, we set $e^- = i$, $e^+ = k$. Each cycle $\omega$ has $k-l$ $A$-edges and $l$ $B$ wedges, and thus 
\begin{align}
&\mathbb{E}_{\lambda,\mu}[Y_{n, k, l} ]  = \frac{\mu^l}{n^{k+l}} \sum_{\omega} \mathbb{E}_{\sigma} \Big[ \sum_{z_e \in \{d, \lambda \sqrt{d} \sigma_{e^-} \sigma_{e^+}\}: e \in E_{\omega,A}} \prod_{e \in E_{\omega, A}} z_{e} \prod_{e' \in E_{\omega, B}} \sigma_{e'^-} \sigma_{e'^+}  \Big]
\end{align}
Observe that we have a non-zero contribution in the sum above if and only if $z_e = \lambda \sqrt{d} \sigma_{e^-} \sigma_{e^+}$ for all $e \in E_{\omega, A}$; in this case, each term contributes $(\lambda \sqrt{d})^{k-l}$. 
Note that the number of $k$ cycles with $l$ $B$-wedges is $\frac{1}{2k} {k \choose l} n^k p^l$, and thus, upon simplification, 
\begin{align}
\mathbb{E}_{\lambda, \mu}[Y_{n,k,l}] = (1+o(1)) \frac{1}{2k} {k \choose l}  \Big( \frac{\mu}{\gamma} \Big)^l (\lambda \sqrt{d})^{k-l}. \nonumber 
\end{align} 
We turn to the calculation of the variance under $\mathrm{H}_1$. Note that under $\mathrm{H}_1$, 
$B_{ij} = X_{ij} + Z_{ij}$, where $X_{ij} = \sqrt{\frac{\mu}{n}} \cdot \sigma_i u_j$. Thus we have, 
\begin{align}
Y_{n, k, l} &= \frac{1}{n^l}\sum_{\omega} \prod_{e_1\in E_{\omega,A}} A_{e_1} \prod_{e_2\in E_{\omega,B}} B_{e_2} \nonumber \\
&=  \frac{1}{n^l}\sum_{\omega} \prod_{e_1\in E_{\omega,A}} A_{e_1} \prod_{e_2\in E_{\omega,B}}  (X_{e_2} + Z_{e_2}) \nonumber \\
&= T_1 + T_2 + T_3, \nonumber 
\end{align}
where 
\begin{align}
T_1 &= \frac{1}{n^l}\sum_{\omega} \prod_{e_1\in E_{\omega,A}} A_{e_1} \prod_{e_2\in E_{\omega,B}} Z_{e_2}, \nonumber \\
T_2 &= \frac{1}{n^l}\sum_{\omega} \prod_{e_1\in E_{\omega,A}} A_{e_1} \prod_{e_2\in E_{\omega,B}} X_{e_2}, \nonumber \\
T_3 &= Y_{n,k,l} - T_1 - T_3. \nonumber 
\end{align}
The term $T_1$ can be analyzed exactly as under $\mathrm{H}_0$, and the same arguments will show that $T_1\overset{d}{\to} \mc N(0, \frac{1}{2k}\binom{k}{l}\frac{d^{k-l}}{\gamma^l})$. We will establish that $T_2 \stackrel{P}{\to} \frac{1}{2k}\binom{k}{l}\frac{(\lambda\sqrt{d})^{k-l}\mu^l}{\gamma^l}$, and that $T_3 \stackrel{P}{\to} 0$.

First, we will consider $T_3$.  Note that $\mathbb{E}_{\lambda, \mu}[T_3]=0$, and thus it suffices to establish that $\mathbb{E}_{\lambda, \mu}[T_3^2] = o(1)$ as $n \to \infty$. Now, we can express $T_3 = \sum_{\omega} V_{n,k,l,\omega}$, where 
\begin{align}
V_{n,k,l,\omega} = \frac{1}{n^l} \prod_{e \in E_{\omega,A}} A_{e} \sum_{E_{\omega,f} \subsetneq E_{\omega,B}} \prod_{e \in E_{\omega,f}} X_{e} \prod_{e \in E_{\omega,B} \backslash E_{\omega,f}} Z_e. \nonumber 
\end{align}
Therefore, 
\begin{align}
&\mathbb{E}_{\lambda, \mu}[T_3^2] = \sum_{\omega_1, \omega_2} \mathbb{E}_{\lambda, \mu}[V_{n,k,l,\omega_1} V_{n,k,l,\omega_2}]  \nonumber \\ 
&:= \sum_{\omega_1, \omega_2} \sum_{E_{\omega_1,f} \subsetneq E_{\omega_1,B}, E_{\omega_2,f} \subsetneq E_{\omega_2,B}}  \mathbb{E}_{\lambda, \mu}[V_{n,k,l,\omega_1, E_{f,\omega_1}} V_{n,k,l,\omega_2, E_{f,\omega_2}}], \nonumber 
\end{align} 
where we define 
\begin{align}
V_{n,k,l,\omega,E_{f,\omega}} := \frac{1}{n^l} \prod_{e \in E_{\omega, A}} A_{e} \prod_{e \in E_{\omega,f}} X_e \prod_{e \in E_{\omega,B} \backslash E_{\omega,f} } Z_e. \nonumber 
\end{align}
This implies $\mathbb{E}_{\lambda, \mu}[ V_{n,k,l,\omega_1,E_{f,\omega_1}} V_{n,k,l,\omega_2,E_{f,\omega_2}} ]$ is zero unless $E_{\omega_1, B} \backslash E_{\omega_1,f} = E_{\omega_2, B} \backslash E_{\omega_2, f}$. Given $\omega_1, \omega_2$, the terms which affect the contribution by powers of $n$ are the edges in $E_{\omega_1, A} \cap E_{\omega_2,A}$. The dominant contribution arises from $\omega_1, \omega_2$ such that $| E_{\omega_1, A} \cap E_{\omega_2, A}| =0$---this follows using the same reasoning used to identify the dominant order of the variance under $\mathrm{H}_0$. Further, note that overlaps in the edges in $E_{\omega_1, f}$ and $E_{\omega_2,f}$ affect the expectation, but only to constant order. For a pair $(\omega_1, \omega_2)$ satisfying these conditions 
\begin{align}
\mathbb{E}_{\lambda, \mu}[V_{n,k,l,\omega_1, E_{f,\omega_1}} V_{n,k,l,\omega_2, E_{f,\omega_2}}] = \frac{1}{n^{2l}} \mathbb{E}_{\lambda, \mu} \Big[ \prod_{e \in E_{\omega_1,A} } A_e \prod_{e \in E_{\omega_2,A} } A_e \prod_{e \in E_{\omega_1,f}} X_e \prod_{e \in E_{\omega_2,f}} X_e \Big]. \nonumber 
\end{align}
As $\omega_1, \omega_2$ are both length $k$ cycles with $k-l$ $A$ edges and $l$ B-edges, and $E_{\omega_1, B} \backslash E_{\omega_1,f} = E_{\omega_2, B} \backslash E_{\omega_2, f}$, we have $|E_{\omega_1,f}| = |E_{\omega_2,f}| := x$. In turn, this implies that there exists $C := C(k,l)>0$ such that 
\begin{align}
\mathbb{E}_{\lambda, \mu}[V_{n,k,l,\omega_1, E_{f,\omega_1}} V_{n,k,l,\omega_2, E_{f,\omega_2}}] &\leq C \frac{1}{n^{2l + x}} \mathbb{E}_{\lambda, \mu}  \Big[ \prod_{e \in E_{\omega_1,A} } A_e \prod_{e \in E_{\omega_2,A} } A_e \Big] \leq C' \frac{1}{n^{2k+x}}, \nonumber 
\end{align} 
where $C' := C'(k,l, \lambda, d)>0$ is a constant independent of $n$. There are only finitely many choices of $E_{\omega_1, f}$ and $E_{\omega_2,f}$, and therefore, for each $\omega_1, \omega_2$, 
\begin{align}
\sum_{E_{\omega_1,f} \subsetneq E_{\omega_1,B}, E_{\omega_2,f} \subsetneq E_{\omega_2,B}}  \mathbb{E}_{\lambda, \mu}[V_{n,k,l,\omega_1, E_{f,\omega_1}} V_{n,k,l,\omega_2, E_{f,\omega_2}}] \leq C'' \frac{1}{n^{2k+x}}, \nonumber 
\end{align} 
for a larger constant $C''$. Finally, we sum over $\omega_1, \omega_2$. If the two cycles intersect on $x$ edges, they have $x+1$ vertices in common. Thus the number of pairs $\omega_1, \omega_2$ with $x$ common edges is $O(n^{2k+ 2l - x-1})$. Summing, we have the conclusion that $\mathbb{E}_{\lambda,\mu} [T_3^2] = o(1)$. 

Finally, we turn to $T_2$. For any $1\leq i_1 < \cdots < i_k \leq n$, we let $\mathbf{i}_{1:k} = (i_1, \cdots ,i_k)$. Similarly, for $1\leq j_1 < \cdots <j_l \leq p$, set $\mathbf{j}_{1:l} = (j_1, \cdots, j_p)$. Finally, given $\mathbf{i}_{1:k}$ and $\mathbf{j}_{1:l}$, let $\mathcal{C}(\mathbf{i}_{1:k}, \mathbf{j}_{1:l})$ denote the set of cycles with $k-l$ $A$-edges and $l$ $B$-wedges on the chosen vertices. 
Armed with this notation, we observe that 
\begin{align}
T_2 &= \frac{1}{n^l} \sum_{\omega} \prod_{e_1 \in E_{\omega,A}} A_{e_1} \prod_{e_2 \in E_{\omega,B}} X_{e_2} \nonumber \\
&= \frac{1}{n^l} \sum_{\mathbf{i}_{1:k} , \mathbf{j}_{1:l}} \sum_{\omega \in \mathcal{C}(\mathbf{i}_{1:k}, \mathbf{j}_{1:l})} \prod_{e_1 \in E_{\omega,A}} A_{e_1} \prod_{e_2 \in E_{\omega,B}} X_{e_2}. \nonumber
\end{align}
Given $\mathbf{i}_{1:k}$, let $\mathcal{C}(\mathbf{i}_{1:k})$ denote all length $k$ cycles on the vertices $i_1, i_2, \cdots, i_k$, with $(k-l)$ edges colored to be of type $A$, and the remaining edges colored to be of type $B$. For any edge in the cycle, let $t(e) \in \{A,B\}$ denote its type. 
This implies 
\begin{align}
T_2 &= \frac{1}{n^l}  \Big(  \sum_{\mathbf{i}_{1:k}} \sum_{\omega \in \mathcal{C}(\mathbf{i}_{1:k})} \prod_{t(e) = A} A_e \prod_{t(e) = B} \sigma_{e^-} \sigma_{e^+}  \Big) \Big( \Big( \frac{\mu}{n} \Big)^l \sum_{\mathbf{j}_{1:l}} \prod_{h=1}^{l} u_{j_h}^2 \Big). \nonumber 
\end{align} 
Note that for fixed $\mb j_{1:l}$, $u_{jh}^2$ are independent with mean $1$, so law of large numbers gives 
\begin{align} 
\frac{1}{n^l} \sum_{\mathbf{j}_{1:l}} \prod_{h=1}^{l} u_{j_h}^2 \stackrel{P}{\to} 1. \nonumber 
\end{align}
Thus it suffices to control the other term. In particular, it suffices to show that 
\begin{align}
\mathrm{Var}_{\lambda,\mu} \Big(\frac{1}{n^l} \Big(  \sum_{\mathbf{i}_{1:k}} \sum_{\omega \in \mathcal{C}(\mathbf{i}_{1:k})} \prod_{t(e) = A} A_e \prod_{t(e) = B} \sigma_{e^-} \sigma_{e^+}  \Big)   \Big) = o(1) \nonumber 
\end{align} 
as $n \to \infty$. This amounts to checking that
\begin{align*}
    \EE[\lambda,\mu]{\frac{1}{n^{2l}} \Big(  \sum_{\mathbf{i}_{1:k}} \sum_{\omega \in \mathcal{C}(\mathbf{i}_{1:k})} \prod_{t(e) = A} A_e \prod_{t(e) = B} \sigma_{e^-} \sigma_{e^+}  \Big)^2} &= \EE[\lambda,\mu]{\frac{1}{n^{l}} \Big(  \sum_{\mathbf{i}_{1:k}} \sum_{\omega \in \mathcal{C}(\mathbf{i}_{1:k})} \prod_{t(e) = A} A_e \prod_{t(e) = B} \sigma_{e^-} \sigma_{e^+}  \Big)}^2 + o(1).
\end{align*}
This turns out to be true because of the same argument for Lemma~\ref{lemma:decoupling} (specifically the discussion after \eqref{eq:lemma2}). This completes the proof. 

\end{proof}

%% file: reconstruction.tex

\section{Weak Recovery under the Threshold}
\label{sec:ubd_below}

In the remaining two sections, we will prove Theorem~\ref{thm:recovery}. In this section, we show the first part, that weak recovery is impossible when $\lambda^2 + \frac{\mu^2}{\gamma} < 1$.  We follow the general proof scheme in \cite{banerjee2018contiguity}. The proof is information theoretic, and the main idea is contained in the following proposition. 
\begin{proposition}
When $\lambda^2 + \frac{\mu^2}{\gamma} < 1$, then for any fixed $r$, and any two configurations $(\sigma_1, ..., \sigma_r), (\tau_1, ..., \tau_r)\in \{\pm 1\}^r$, we have that as $n\to \infty$,
\begin{align*}
    \norm{\mbb P_{\lambda, \mu}(\cdot\mid \sigma_{1:r}) - \mbb P_{ \lambda, \mu }(\cdot\mid \tau_{1:r})}_{\mathrm{TV}} \to 0.
\end{align*}
\end{proposition}

\begin{proof}
The idea is to bound the total variation with a function of the second moment of the likelihood ratios,
\begin{align*}
    L_{\sigma, n} = \frac{\mathrm{d} \mathbb{P}_{\lambda, \mu}}{\mathrm{d}\mathbb{P}_{0,0}}(\cdot\mid \sigma_{1:r}) \ , \ L_{\tau, n} = \frac{\mathrm{d} \mathbb{P}_{\lambda, \mu}}{\mathrm{d}\mathbb{P}_{0,0}}(\cdot\mid \tau_{1:r}),
\end{align*}
and then noting that because $r$ is fixed, the second moments of these likelihood ratios will converge to a limit independent of $\sigma_{1:r}$ and $\tau_{1:r}$, making the upper bound converge to $0$. However, as we have seen in the other arguments, we need some truncation. Let us define the truncated likelihood ratios
\begin{align*}
    \tilde L_{\sigma, n} &= \frac{\EE[\sigma_{-r}, u]{\PP[\lambda, \mu]{\mathbf{A}, \mathbf{B}\mid \sigma_{1:r}, \sigma_{-r}, u}\mathbf{1}(u \in \mathcal{S})}}{\PP[0,0]{\mathbf{A}, \mathbf{B}}} \\ 
    \tilde L_{\tau, n} &= \frac{\EE[\tau_{-r}, u]{\PP[\lambda, \mu]{ \mathbf{A}, \mathbf{B}\mid \tau_{1:r}, \tau_{-r}, u}\mathbf{1}(u \in \mathcal{S})}}{\PP[0,0]{\mathbf{A},  \mathbf{B}}},
\end{align*}
where $\sigma_{-r}, \tau_{-r}\in \{\pm 1\}^{n-r}$ are the other coordinates, and $\mc S = \{\norm{u}\leq 2\sqrt{p}\}$. Define distributions given by these truncated likelihood ratios:
\begin{align*}
    \Q_{\sigma,n}(\Omega\mid \sigma_{1:r}) &= \frac{1}{\PP[n]{\mc S}} \EE[0,0]{\tilde L_{\sigma,n} \one(\Omega)\mid \sigma_{1:r}}\\
    \Q_{\tau,n}(\Omega\mid \tau_{1:r}) &= \frac{1}{\PP[n]{\mc S}} \EE[0,0]{\tilde L_{\tau,n} \one(\Omega)\mid \tau_{1:r}}.
\end{align*}
From the proof of \cite[Proposition 1]{banerjee2018lr}, we know that $\norm{\PP[\lambda, \mu]{\cdot \mid \sigma_{1:r}}  - \Q_{\sigma, n}(\cdot \mid \sigma_{1:r})}_{\mathrm{TV}}$ and $\norm{\PP[\lambda, \mu]{\cdot \mid \tau_{1:r}}  - \Q_{\tau, n}(\cdot \mid \tau_{1:r})}_{\mathrm{TV}}$ both vanish as $n\to\infty$. Thus, to prove the proposition, it suffices to check that $\norm{\Q_{\sigma, n}(\cdot \mid \sigma_{1:r})  - \Q_{\tau, n}(\cdot \mid \tau_{1:r})}_{\mathrm{TV}}\to 0$. Note that
\begin{align*}
    &\norm{\Q_{\sigma, n}(\cdot \mid \sigma_{1:r})  - \Q_{\tau, n}(\cdot \mid \tau_{1:r})}_{TV}
    = \frac{1}{\PP[n]{\mc S}} \EE[0,0]{ |\tilde L_{\sigma, n} - \tilde L_{\tau,n}|}
    \leq \frac{1}{\PP[n]{\mc S}}\EE[0,0]{(\tilde L_{\sigma, n} - \tilde L_{\tau, n})^2}^{\frac 1 2},
\end{align*}
where we have used the  Cauchy Schwarz inequality. Now, we have, 
\begin{align*}
    &\EE[0,0]{(\tilde L_{\sigma, n} - \tilde L_{\tau, n})^2}\\
    &= \E_{0,0}\bigg[\frac{1}{\PP[0,0]{\mathbf{A}, \mathbf{B}}^2} \E_{\sigma_{-r}, \tau_{-r}, u, v}\Big\{\big( \PP[ \lambda, \mu]{ \mathbf{A},  \mathbf{B} \mid \sigma_{1:r}, \sigma_{-r}, u} \PP[ \lambda, \mu]{\mathbf{A}, \mathbf{B}\mid \sigma_{1:r}, \tau_{-r}, v}\\
    &\ \ \ \ \ \ \ \ + \PP[\lambda, \mu]{ \mathbf{A}, \mathbf{B}\mid \tau_{1:r}, \sigma_{-r}, u} \PP[ \lambda, \mu]{ \mathbf{A}, \mathbf{B}\mid \tau_{1:r}, \tau_{-r}, v}\\
    &\ \ \ \ \ \ \ \ - 2\PP[\lambda, \mu]{ \mathbf{A}, \mathbf{B}\mid \sigma_{1:r}, \sigma_{-r}, u} \PP[\lambda, \mu]{ \mathbf{A}, \mathbf{B}\mid \tau_{1:r}, \tau_{-r}, v}\big)\mathbf{1}(u,v \in \mathcal{S})\Big\}\bigg].
\end{align*}
Thus we just have to prove that the quantity
\begin{align*}
    \EE[0,0]{\frac{\PP[ \lambda, \mu]{ \mathbf{A}, \mathbf{B}\mid \sigma_{1:r}, \sigma_{-r}, u} \PP[\lambda, \mu]{ \mathbf{A}, \mathbf{B}\mid \tau_{1:r}, \tau_{-r}, v}}{\PP[0,0]{ \mathbf{A}, \mathbf{B}}^2}\mathbf{1}(u, v \in \mathcal{S})}
\end{align*}
has a limit which is independent of $\sigma_{1:r}$ and $\tau_{1:r}$. But we know that this is true because of the second moment calculations in Section \ref{sec:detection}, so we are done.
\end{proof}

Then the impossibility of reconstruction follows from some technical calculations. The proof of the next two results follow directly from Proposition~6.2 and Theorem~2.2 of \cite{banerjee2018contiguity} respectively. 
\begin{proposition}
Let $\lambda^2 + \frac{\mu^2}{\gamma}<1$. Let $S\subset[n]$ be such that $|S| = r$, with $r$ finite and fixed, and let $u\in [n]$ be a single index such that $u\not\in S$. Then, as $n\to\infty$, we have that
\begin{align*}
    \mathbb{E}_{\lambda,\mu} [\norm{\mbb P_{ \lambda, \mu}(\sigma_u \mid \mathbf{A}, \mathbf{B}, \sigma_S) - \mbb P_{ 0, 0}(\sigma_u)}_{\mathrm{TV}}\mid \sigma_S] \to 0.
\end{align*}
\end{proposition}

\begin{theorem}
If $\lambda^2 + \frac{\mu^2}{\gamma}<1$, then reconstruction is impossible, i.e. let the overlap be defined as
\begin{align*}
    ov(\sigma, \tau) = \frac 1 n \sum_{i=1}^n \sigma_i \tau_i - \p{\frac 1 n \sum_{i=1}^n \sigma_i}\p{\frac 1 n \sum_{i=1}^n\tau_i},
\end{align*}
then for any estimator $\hat\sigma( \mathbf{A},  \mathbf{B} )\in \{\pm 1\}^n$, we have that
\begin{align*}
    ov(\sigma, \hat\sigma)\overset{p}{\to}0.
\end{align*}
Because $\frac 1 n \sum_{i=1}^n \sigma_i \overset{p}{\to} 0$, and $\frac 1 n \sum_{i=1}^n \hat\sigma_i$ is bounded, we see that $\frac 1 n \inner{\sigma, \hat\sigma}\overset P \to 0$, so weak recovery is impossible.
\end{theorem}

\section{Weak Recovery with Self Avoiding Walks}
\label{sec:ubd}

In this section we finish the proof of Theorem~\ref{thm:recovery}, by showing that weak recovery is possible whenever $\lambda^2 + \frac{\mu^2}{\gamma} > 1$. Because weak recovery is possible as soon as either $\lambda^2>1$ \cite{mossel2013proof,massoulie2014community} or $\frac{\mu^2}{\gamma}>1$ \cite{baik2005phase}, we only need to consider the case that $\lambda^2$, $\frac{\mu^2}{\gamma} < 1$. We will construct an estimator $\hat\sigma$ that is computable in quasi-polynomial time. We use a strategy introduced in \cite{hopkins2017bayesian} in the context of community detection in the block model. Under this approach, one seeks to design an appropriate set of ``low-degree" polynomials in the data $(\mathbf{A}, \mb B)$, and recover the signals based on these polynomials. We note that in the specific context of community detection, this approach was already latent in the approach of \cite{massoulie2014community} and \cite{bordenave2015non}, based on self-avoiding/non-backtracking walks.

We seek to calculate a polynomial $P(\mathbf{A},\mathbf{B})$, which estimates $\sigma \sigma^{\mathrm{T}}$. Formally, suppose we had an estimator satisfying 
\begin{align}
\mathbb{E}_{\lambda,\mu}\Big[ \langle P(\mathbf{A}, \mb B) ,  \sigma \sigma^{\mathrm{T}} \rangle \Big] \geq \delta \mathbb{E}_{\lambda,\mu}\Big[ \| P(\mathbf{A}, \mb B ) \|_{F}^2 \Big]^{\frac 1 2} \label{eq:desiderata}
\end{align}
for some universal constant $\delta>0$. Then \cite[Theorem 1]{hopkins2017bayesian} implies that there exists $\delta' = \delta'(\delta)$, and an estimator $\hat{\sigma}$ such that 
\begin{align}
\frac{1}{n^2}\mathbb{E}_{\lambda,\mu}[\langle \sigma, \hat{\sigma} \rangle^2] \geq \delta'. \nonumber 
\end{align}
This ensures weak recovery in our setting. 
%
To construct the estimator $\hat{\sigma}$, from the matrix $P(\mb A, \mb B)$ constructed above, we compute a matrix $\Sigma$ with minimum Frobenius norm that satisfies the following constraints:
\begin{align*}
\diag(\Sigma) &= \one \\
\frac{\inner{P(\mb A, \mb B), \Sigma}}{\norm{P(\mb A, \mb B)}_F \cdot n}  &\geq \delta'\\
\Sigma &\succeq 0
\end{align*}
and then output the vector $\hat \sigma\in \{\pm 1\}^n$ obtained by taking coordinate-wise signs of a centered Gaussian vector with covariance $\Sigma$. 

\begin{lemma}
\label{lemma:weak_recovery}
The estimator $\hat{\sigma}$ achieves weak recovery whenever $\lambda^2 + \frac{\mu^2}{\gamma} >1$. 
\end{lemma}
\begin{proof}
The proof of weak recovery follows immediately from the proof of  \cite[Lemma~3.5]{hopkins2017bayesian}. 
\end{proof}

\begin{remark}
This estimator takes $n^{O(\log n)/\text{poly}(\delta)}$ time to compute, as we use certain Self Avoiding Walks (SAWs) of length $\Theta(\log n)$. The running time can be improved to $n^{\text{poly}(1/\delta)}$ by the idea of color coding, and more discussions on this improvement can be found in \cite[Section~2.5]{hopkins2017bayesian}.
\end{remark}

In the remainder of the section, we will construct an estimator and establish \eqref{eq:desiderata}.

\subsection{Self Avoiding Walks}  

It remains to construct the polynomial $P(\mathbf{A}, \mb B)$. To this end, we will use self avoiding walks on the underlying factor-graph. First, for $i_1, i_2 \in [n]$ and $j \in [p]$, we define 
\begin{align}
\hat{A}_{i_1, i_2} &= \frac{2n}{a-b} \Big(A_{ij} - \frac{a+b}{2n} \Big), \nonumber \\
\hat{B}^{j}_{i_1,i_2} &= \frac{n}{\mu}  B_{i_1,j} B_{i_2,j} = \frac{n}{\mu} \Big( \sqrt{\frac{\mu}{n}} \sigma_{i_1} u_j + Z_{i_1, j}\Big) \Big(\sqrt{\frac{\mu}{n}} \sigma_{i_2} u_j + Z_{i_2,j}  \Big). \label{eq:weights}
\end{align} 
Direct computation yields that 
\begin{align}
\mathbb{E}_{\lambda,\mu}[\hat{A}_{i_1,i_2} | \sigma] = \sigma_{i_1} \sigma_{i_2}, \,\,\,\, \mathrm{Var}_{\lambda,\mu}(A_{i_1,i_2}) = \frac{n}{\lambda^2}, \nonumber \\
\mathbb{E}_{\lambda,\mu}[\hat{B}^j_{i_1,i_2} | \sigma] = \sigma_{i_1} \sigma_{i_2}, \,\,\,\, \mathrm{Var}_{\lambda,\mu}(\hat{B}^j_{i_1,i_2}) = \frac{np}{\mu^2/\gamma}. \nonumber
\end{align}
Recall the factor graph, as shown in Figure~\ref{fig:factor}. For $i_1, i_2 \in [n]$, we will associate the weight $\hat{A}_{i_1, i_2}$ to the $A$ edge $\{i_1, i_2\}$. Similarly, for $j\in [p]$, we associate the weight $\hat B^j_{i_1} = \sqrt{\frac \mu n} B_{i_1, j}$ to the $B$ edge $\{i_1, j\}$, and the weight $\hat{B}^{j}_{i_1,i_2}$ to the $B$ wedge $\{i_1, j, i_2\}$.

Fix $i_1, i_2 \in [n]$, and let $k \geq l \geq 1$ be integers that we will specify later.  Consider a path $\alpha$ on the factor graph that starts at $i_1$ and ends at $i_2$, which contains $k-l$ $A$ type edges, and $l$ $B$ type wedges. We will require that the $A$-type edges on the path are ``self-avoiding", i.e., no edge of type $A$ occurs more than once. Further, if $j_1, \cdots, j_l$ denote the vertices in $V_2$ which lie on the path $\alpha$, we will require that these vertices are distinct. Let $\mathcal{L}(i_1, i_2, k, l)$ denote the set of all such paths $\alpha$, for any given $i_1, i_2 \in [n]$, and $k \geq l$. Given any path $\alpha$, construct a polynomial on entries of $(\mb A, \mb B)$ by 
\begin{align}
p_{\alpha} = \prod_{e \in \alpha} \mathrm{weight}(e), \nonumber 
\end{align} 
where $\mathrm{weight}(i_1, i_2)$ is $\hat A_{i_1, i_2}$ and $\hat B^j_{i_1. i_2}$ when the edge $(i_2, i_2)$ is of type $A$ and $B$ respectively. Direct computation yields that 
\begin{align}
\mathbb{E}_{\lambda,\mu}[p_{\alpha} | \sigma] = \sigma_{i_1} \sigma_{i_2}, \,\,\,\,\, \mathrm{Var}_{\lambda,\mu}(p_{\alpha}) = \Big( \frac{n}{\lambda^2} \Big)^{k-l} \Big( \frac{np}{\mu^2/\gamma} \Big)^l (1+o(1)). \label{eq:paths_prop}
\end{align}
Thus we see that $p_\alpha$ is an unbiased estimator for $\sigma_{i_1}\sigma_{i_2}$, but its large variance renders it useless on its own. Fortunately, there are many paths $\alpha\in \mc L(i_1, i_2, k, l)$, and we might hope that we can reduce the variance by averaging over polynomials from different paths as follows 
\begin{align}
P_{i_1, i_2} (\mathbf{A}, \mb B) = \frac{1}{|\mc L(i_1, i_2, k, l)|}\sum_{\mathcal{L}(i_1, i_2, k, l)} p_{\alpha} . \nonumber
\end{align} 
Note that $P_{i_1, i_2}$ is still an unbiased estimator for $\sigma_{i_1}\sigma_{i_2}$. Finally, we set $P(\mathbf{A}, \mb B) = \{ P_{i_1, i_2}(\mathbf{A}, \mb B) : 1\leq  i_1< i_2 \leq n\}$ to be the estimator for the matrix $\sigma\sigma^\top$. 

\noindent
It remains to check \eqref{eq:desiderata} whenever $\lambda^2 + \frac{\mu^2}{\gamma} >1$. We establish this in the next lemma.

\begin{lemma}
Assume $d>1$ and $\lambda^2 + \frac{\mu^2}{\gamma}>1+ \varepsilon$, for some $\varepsilon>0$. Then there exists universal constants $C>0$ and $c>0$ such that setting $k = C\log n/\varepsilon^c$ and $l := l(k, \lambda, \mu, \gamma) \leq k$ such that 
\begin{align}
\mathbb{E}_{\lambda,\mu}\Big[ \langle P(\mathbf{A}, \mb B ) , \sigma \sigma^{\mathrm{T}} \rangle \Big] \geq \delta \mathbb{E}_{\lambda,\mu}\Big[ \| P(\mathbf{A}, \mb B )  \|_F^2 \Big]^{\frac 1 2}
\end{align}
for some $\delta := \delta(\varepsilon, C, c, \lambda, \mu, \gamma)>0$.
\end{lemma}

To prove this lemma, we will show that an entry-wise version holds: 
\begin{align}
	\mathbb{E}_{\lambda,\mu}\Big[P_{i_1, i_2}(\mb A, \mb B) \cdot \sigma_{i_1}\sigma_{i_2}\Big]  \geq \delta \mathbb{E}_{\lambda,\mu}\Big[P_{i_1, i_2}(\mb A, \mb B)^2\Big]^{\frac{1}{2}}.\label{eq:cor1}
\end{align}
By construction, we have that $\mathbb{E}_{\lambda,\mu}[P_{i_1, i_2}(\mb A, \mb B) \ \mid \ \sigma_{i_1}\sigma_{i_2}] = \sigma_{i_1}\sigma_{i_2}$. As a result, it is easy to check that \eqref{eq:cor1} is implied by the following:
\begin{align}\label{eq:uncorrelated}
\mathbb{E}_{\lambda,\mu}[(\sigma_{i_1}\sigma_{i_2})^2] \cdot \sum_{\alpha, \beta \in \mc L(i_1, i_2, k, l)} \mathbb{E}_{\lambda,\mu}[p_\alpha \cdot p_\beta] \leq \frac{1}{\delta^2} \cdot \sum_{\alpha, \beta\in \mc L(i_1, i_2, k, l)} \mathbb{E}_{\lambda,\mu}[ p_\alpha \cdot (\sigma_{i_1}\sigma_{i_2})] \mathbb{E}_{\lambda,\mu}[ p_\beta \cdot (\sigma_{i_1}\sigma_{i_2})].
\end{align}
Intuitively, this inequality is saying that the correlation between different self avoiding walks is not too large. The right hand side of the inequality is easy to control: note that $\mathbb{E}_{\lambda,\mu}[ p_\alpha \cdot (\sigma_{i_1}\sigma_{i_2})] = \mathbb{E}_{\lambda,\mu}[(\sigma_{i_1}\sigma_{i_2})^2] = 1$, so the right hand side is equal to $\delta |\mc L(i_1, i_2, k, l)|^2$. We see that this is given by
\begin{align}\label{eq:RHS_lower_bdd}
|\mc L(i_1, i_2, k, l)|^2 = (1+o(1))\binom{k}{l}^2n^{2(k-1)}p^{2l}.
\end{align}
Thus we are left to control correlation between $p_\alpha$ and $p_\beta$, given on the left hand side of \eqref{eq:uncorrelated}.

\subsection{Correlation of SAWs: Proof of \eqref{eq:uncorrelated}}

We want an upper bound on the left hand side of \eqref{eq:uncorrelated}, so we need to control the correlation $\EE{p_\alpha p_\beta}$, for paths $\alpha, \beta\in \mc L(i_1, i_2, k, l)$. For this, we have to keep track of the number of intersections. Let $\tilde{a}$ be the number of $\hat{A}$-edge intersections, and $\tilde{b}$ be the number of $\hat{B}$-edge intersections. In particular, we must have $\tilde{a}\leq k-l$, and $\tilde{b} \leq l$.

Note that the contribution to correlation depends only on how many intersections there are, and does not depend on the other edges of $\alpha$ and $\beta$. Computations show that each $\hat{A}$-edge intersection contributes a factor of $O(n/\lambda^2)$, and each $\hat{B}$-edge intersection contributes a factor of $O(n/\mu)$, so the total contribution of such an intersection would simply be
\begin{equation*}
O(1)\p{\frac{n}{\lambda^2}}^{\tilde{a}}\p{\frac{n}{\mu}}^{\tilde{b}}.
\end{equation*}
Now we calculate the number of pairs $\alpha, \beta\in \mc L(i_1, i_2, k, l)$ that intersect on $\tilde a$ $A$ edges and $\tilde b$ $B$ edges. First we will calculate the number of pairs $\alpha$ and $\beta$ that intersect on the smallest number of vertices, given the number of edge intersections. The smallest numbers of vertex intersections are $\tilde{a} + \lfloor\tilde{b}/2\rfloor$ $V_1$ vertices, and $\lceil\tilde{b}/2\rceil$ $V_2$ vertices. This is achieved when both paths, $\alpha$ and $\beta$, begin with $k-l$ consecutive $A$ edges, and end with $l$ $B$ wedges, and intersect on the first $\tilde{a}$ $\hat{A}$-edges, as well as the last $\tilde{b}$ $\hat{B}$-edges. In this case, there are $2(k-1) - (\tilde{a} + \lfloor\tilde{b}/2\rfloor)$ free $V_1$ vertices for the two paths to choose, and $b_1 + b_2 - \lceil\tilde{b}/2\rceil$ free $V_2$ vertices for the paths to choose, which means that the total number of such pairs is

\begin{equation*}
n^{2(k-1) - \tilde{a}-\lfloor\tilde{b}/2\rfloor}p^{2l - \lceil\tilde{b}/2\rfloor}.
\end{equation*}
A similar calculation shows that number of pairs of paths that intersects on the least number of vertices contributes the leading order term. This is because if two paths intersected on $r$ more vertices, then the total number of paths will decrease by a factor of $l^{O(r)}n^{-O(r)}$. Thus the total contribution in the correlation (in \eqref{eq:uncorrelated}) is:

\begin{equation*}
\begin{aligned}
&O(1)n^{2(k-l-1) - \tilde{a}-\lfloor\tilde{b}/2\rfloor}p^{2l - \lceil\tilde{b}/2\rfloor}\p{\frac{n}{\lambda^2}}^{\tilde{a}}\p{\frac{n}{\mu}}^{\tilde{b}} \\
&= O(1) n^{2(k-1)}p^{2l} \p{(\lambda^2)^{-\tilde{a}} \p{\frac{\mu^2}{\gamma}}^{-\tilde{b}/2} \gamma^{\lceil\tilde{b}/2\rceil - \tilde{b}/2}}. 
\end{aligned}
\end{equation*}
Thus, summing over possible number of intersections, we have that:
\begin{align}\label{eq:LHS_upper_bdd}
\begin{aligned}
&\ \ \ \ \sum_{\alpha, \beta \in \mc L(i_1, i_2, k, l)}  \mathbb{E}_{\lambda,\mu}[p_\alpha \cdot p_\beta]\\
 &\leq O(1)n^{2(k-1)}p^{2l} \sum_{\tilde a\leq l-m}\sum_{\tilde b\leq 2m}\p{(\lambda^2)^{-\tilde{a}}\p{\frac{\mu^2}{\gamma}}^{-\tilde{b}/2} \gamma^{\lceil\tilde{b}/2\rceil - \tilde{b}/2}}\\
&\leq O(1) n^{2(k-1)}p^{2l} (\lambda^2)^{-(k-l)}\p{\frac{\mu^2}{\gamma}}^{-l}.
\end{aligned}
\end{align}
Thus, to achieve the bound in \eqref{eq:uncorrelated}, we just need to show that \eqref{eq:RHS_lower_bdd} really is an upper bound of \eqref{eq:LHS_upper_bdd}. That amounts to showing
\begin{align*}
\binom{k}{l}^2 \gtrsim (\lambda^2)^{-(k-l)}\p{\frac{\mu^2}{\gamma}}^{-l}.
\end{align*}
Let us now choose $\frac{l}{k} = \frac{\mu^2/\gamma}{\lambda^2 + \mu^2/\gamma}$, and hence $\frac{k-l}{k} = \frac{\lambda^2}{\lambda^2 + \mu^2/\gamma}$. Note that we are assuming $\lambda^2 + \mu^2/\gamma > 1$, so $\lambda^2 > \frac{k-l}{k}$, and $\mu^2/\gamma > \frac{l}{k}$. As a result, we see that
\begin{align*}
&\ \ \ \ \binom{k}{l}^2 = \exp\p{- 2l\log\frac{l}{k} - 2(k-l)\log\frac{k-l}{k} + o(1)}\\
&\gtrsim \exp\p{- l\log\frac{\mu^2}{\gamma} - (k-l)\log\lambda^2} = (\lambda^2)^{-(k-l)}\p{\frac{\mu^2}{\gamma}}^{-l}
\end{align*}
which is exactly what we wanted to show.